\documentclass[11pt]{article}
\usepackage[backref=page]{hyperref}
\usepackage{makeidx}
\usepackage{amsmath,amssymb,amsfonts,amsthm}
\usepackage{subcaption}
\usepackage{graphicx}
\usepackage{fullpage}
\usepackage{color}
\usepackage{wrapfig}
\usepackage{tikz}
\usetikzlibrary{decorations.pathreplacing}
\usepackage{setspace}
\usepackage{algorithm}
\usepackage[noend]{algpseudocode}
\usepackage[framemethod=tikz]{mdframed}
\usepackage{xspace}
\usepackage{pgfplots}
\usepackage{framed}
\usepackage{thmtools}
\usepackage{thm-restate}
\pgfplotsset{compat=1.5}

\newtheorem{theorem}{Theorem}[section]
\newtheorem{corollary}[theorem]{Corollary}
\newtheorem{lemma}[theorem]{Lemma}

\newtheorem{definition}[theorem]{Definition}
\newtheorem{remark}[theorem]{Remark}

\newtheorem{fact}[theorem]{Fact}

\newenvironment{proofof}[1]{\begin{trivlist} \item {\bf Proof
#1:~~}}
  {\qed\end{trivlist}}

\newcommand{\namedref}[2]{\hyperref[#2]{#1~\ref*{#2}}}
\newcommand{\thmlab}[1]{\label{thm:#1}}
\newcommand{\thmref}[1]{\namedref{Theorem}{thm:#1}}
\newcommand{\lemlab}[1]{\label{lem:#1}}
\newcommand{\lemref}[1]{\namedref{Lemma}{lem:#1}}

\newcommand{\corlab}[1]{\label{cor:#1}}

\newcommand{\seclab}[1]{\label{sec:#1}}

\newcommand{\applab}[1]{\label{app:#1}}
\newcommand{\appref}[1]{\namedref{Appendix}{app:#1}}
\newcommand{\factlab}[1]{\label{fact:#1}}
\newcommand{\factref}[1]{\namedref{Fact}{fact:#1}}
\newcommand{\remlab}[1]{\label{rem:#1}}
\newcommand{\remref}[1]{\namedref{Remark}{rem:#1}}
\newcommand{\figlab}[1]{\label{fig:#1}}
\newcommand{\figref}[1]{\namedref{Figure}{fig:#1}}
\newcommand{\alglab}[1]{\label{alg:#1}}
\renewcommand{\algref}[1]{\namedref{Algorithm}{alg:#1}}

\newcommand{\deflab}[1]{\label{def:#1}}
\newcommand{\defref}[1]{\namedref{Definition}{def:#1}}

\newcommand{\PPr}[1]{\ensuremath{\mathbf{Pr}\left[#1\right]}}

\newcommand{\Ex}[1]{\ensuremath{\mathbb{E}\left[#1\right]}}

\renewcommand{\O}[1]{\ensuremath{\mathcal{O}\left(#1\right)}}
\newcommand{\tO}[1]{\ensuremath{\tilde{\mathcal{O}}\left(#1\right)}}
\newcommand{\eps}{\varepsilon}
\def \identify    {\mdef{\textsc{Identify}}}
\def \heavyhitters    {\mdef{\textsc{HeavyHitters}}}
\def \countsketch    {\mdef{\textsc{CountSketch}}}
\def \bptree    {\mdef{\textsc{BPTree}}}
\def \counthh    {\mdef{\textsc{CountHH}}}
\def \counthhsub    {\mdef{\textsc{CountHHShort}}}
\def \id    {\mdef{\mathsf{id}}}

\def \disjinfty    {\mdef{\textsc{DisjInfty}}}

\def \IC    {\mdef{\mathsf{IC}}}
\def \CIC    {\mdef{\mathsf{CIC}}}

\def \calA    {\mdef{\mathcal{A}}}

\def \e    {\mdef{\mathbf{e}}}

\def \u    {\mdef{\mathbf{u}}}

\def \v    {\mdef{\mathbf{v}}}
\def \x    {\mdef{\mathbf{x}}}

\newcommand{\mdef}[1]{{\ensuremath{#1}}\xspace}  

\DeclareMathOperator*{\polylog}{polylog}
\DeclareMathOperator*{\poly}{poly}
\DeclareMathOperator*{\median}{median}
\DeclareMathOperator*{\Var}{Var}
\DeclareMathOperator*{\embed}{embed}

\newcommand{\ignore}[1]{}

\newif\ifnotes\notestrue 
\ifnotes
\newcommand{\samson}[1]{\textcolor{purple}{{\bf (Samson:} {#1}{\bf ) }} \marginpar{\tiny\bf
             \begin{minipage}[t]{0.5in}
               \raggedright S:
            \end{minipage}}}            							
\else
\newcommand{\samson}[1]{}
\fi

\hypersetup{
     colorlinks   = true,
     citecolor    = mahogany,
		 linkcolor		= darkpastelgreen
}

\makeatletter
\renewcommand*{\@fnsymbol}[1]{\textcolor{mahogany}{\ensuremath{\ifcase#1\or *\or \dagger\or \ddagger\or
 \mathsection\or \triangledown\or \mathparagraph\or \|\or **\or \dagger\dagger
   \or \ddagger\ddagger \else\@ctrerr\fi}}}
\makeatother

\providecommand{\email}[1]{\href{mailto:#1}{\nolinkurl{#1}\xspace}}

\definecolor{mahogany}{rgb}{0.75, 0.25, 0.0}
\definecolor{darkpastelgreen}{rgb}{0.01, 0.75, 0.24}
\begin{document}

\title{Separations for Estimating Large Frequency Moments on Data Streams}
\author{David P. Woodruff\thanks{Carnegie Mellon University. 
E-mail: \email{dwoodruf@cs.cmu.edu}}\\
\and
Samson Zhou\thanks{Carnegie Mellon University. 
E-mail: \email{samsonzhou@gmail.com}}
}
\date{\today}

\maketitle

\begin{abstract}
We study the classical problem of moment estimation of an underlying vector whose $n$ coordinates are implicitly defined through a series of updates in a data stream. We show that if the updates to the vector arrive in the random-order insertion-only model, then there exist space efficient algorithms with improved dependencies on the approximation parameter $\varepsilon$. In particular, for any real $p > 2$, we first obtain an algorithm for $F_p$ moment estimation using $\tilde{\mathcal{O}}\left(\frac{1}{\varepsilon^{4/p}}\cdot n^{1-2/p}\right)$ bits of memory. Our techniques also give algorithms for $F_p$ moment estimation with $p>2$ on arbitrary order insertion-only and turnstile streams, using $\tilde{\mathcal{O}}\left(\frac{1}{\varepsilon^{4/p}}\cdot n^{1-2/p}\right)$ bits of space and two passes, which is the first optimal multi-pass $F_p$ estimation algorithm up to $\log n$ factors. Finally, we give an improved lower bound of $\Omega\left(\frac{1}{\varepsilon^2}\cdot n^{1-2/p}\right)$ for one-pass insertion-only streams. Our results separate the complexity of this problem both between random and non-random orders, as well as one-pass and multi-pass streams.  
\end{abstract}

\section{Introduction}
The efficient computation of statistics has emerged as an increasingly important goal for large databases storing information generated from financial markets, internet traffic, IoT sensors, scientific observations, etc. 
The one-pass streaming model formally defines an implicit and underlying dataset through sequential updates that arrive one at a time and describe the evolution of the dataset over time. 
The goal is to aggregate or approximate some statistic of the input data using space that is sublinear in the size of the input. 

The frequency moment estimation problem is fundamental to data streams. 
Since the celebrated paper of Alon, Matias, and Szegedy~\cite{AlonMS99}, the frequency moment estimation problem has been a central problem in the streaming model; more than two decades of research~\cite{AlonMS99,ChakrabartiKS03,Bar-YossefJKS04,Woodruff04,IndykW05,Indyk06,Li08,KaneNW10,KaneNPW11,Ganguly11,BravermanO13,BravermanKSV14,BlasiokDN17,BravermanVWY18, GangulyW18,WoodruffZ20} has studied the space or time complexity of this problem. 

We first consider the \emph{insertion-only} model, where updates take the form $u_1,\ldots,u_m$ in a stream of length $m$ and each update $u_t$ is in $[n] = \{1, 2, \ldots, n\}$ for $t\in[m]$. 
We assume for simplicity that $m\le\poly(n)$. 
The updates implicitly define a frequency vector $f\in\mathbb{R}^n$ so that each update effectively increases a coordinate of $f$ in the sense that $f_i=|\{t\,:\,u_t=i\}|$ for each $i\in[n]$. 
Given $p,\eps>0$, the $F_p$ moment estimation problem is to approximate $F_p=\sum_{i\in[n]}(f_i)^p$ within a $(1\pm\eps)$ factor.  
The complexity of this problem differs greatly for the range of $p$. 
For $p>2$, \cite{Bar-YossefJKS04,ChakrabartiKS03} showed that even for the streaming model, the space usage for $F_p$-estimation requires polynomial factors in $n$ and $m$, whereas polylogarithmic factors are achievable for $p\le 2$~\cite{AlonMS99,Indyk06,Li08,KaneNW10,KaneNPW11,BlasiokDN17}. 

We initially focus on the frequency moment estimation problem in the \emph{random-order} model, where the set of stream updates to the underlying frequency vector is worst case, but the order of their arrival is uniformly random. 
As usual, the algorithms we initially consider are only permitted a single pass over the stream, though we later relax this constraint to permit adversarial ordering of updates, as well as both positive and negative integer updates (with magnitude bounded by $\poly(n)$) to the coordinates of the frequency vector in the \emph{turnstile} model. 
Random-order streams have been shown to be a natural assumption for problems of sorting and selecting in limited space~\cite{MunroP80} and many other real-world applications~\cite{GuhaM06, GuhaM07, DemaineLM02, ChakrabartiJP08}. 
Interestingly, there has sometimes been a significant qualitative difference between random-order streams and adversarial (or arbitrary) order streams~\cite{KonradMM12,BravermanKSV14,BravermanVWY18}. 
In particular for $F_p$-moment estimation, the best known lower bound is $\Omega\left(\frac{1}{\eps^2}\right)$~\cite{ChakrabartiC016}. 
While the best known upper bound for $p\in(0,2]$ on arbitrary order insertion-only streams is $\O{\frac{1}{\eps^2}\log n}$~\cite{AlonMS99, KaneNW10, KaneNPW11}, Braverman \emph{et. al.}~\cite{BravermanVWY18} gave an algorithm for $p\in(0,2)$ on random-order streams that only used $\tO{\frac{1}{\eps^2}+\log n}$ space, where for a function $g(n, \varepsilon)$, we use $\tO{g(n,\varepsilon)}$ to denote a function bounded by $g(n,\varepsilon) \cdot \polylog(g(n,\varepsilon))$. 
They also gave an algorithm for $F_2$-moment estimation that only uses $\tO{\frac{1}{\eps^2}+\log n}$ space, but requires the assumption that $F_2\ge F_1\cdot\log n$, i.e., the second moment must be a logarithmic factor larger than the length of the stream. 

The above works raise a number of important questions. 
A tantalizing open question, studied extensively in \cite{JayaramW19}, and dating back to the original work of Alon, Matias, and Szegedy \cite{AlonMS99}, is the exact space complexity of $F_p$-moment estimation in insertion-only streams. 
For $p>2$, the best known upper bound for $F_p$-moment estimation on arbitrary order insertion-only streams is the minimum of $\O{\frac{n^{1-2/p}}{\varepsilon^{20}}}$ \cite{BravermanKSV14} and $\tO{\frac{1}{\eps^2}\cdot n^{1-2/p}}$ \cite{Ganguly11,GangulyW18}. 
For insertion-only streams, the best known lower bound is $\tilde{\Omega}\left(\frac{n^{1-2/p}}{\eps^2 \log n} \right )$~\cite{Ganguly12}. 
We summarize these results in \figref{fig:history}.
\begin{figure*}[!htb]
\begin{center}
\begin{tabular}{|c|c|c|}\hline
Reference & Space Complexity & Stream Order \\\hline
\cite{IndykW05, MonemizadehW10, Andoni17} & $\tO{n^{1-2/p}\eps^{-\O{1}}}$ & Arbitrary \\
\cite{BravermanKSV14} & $\O{n^{1-2/p}\eps^{-20}}$ & Arbitrary \\
\cite{AndoniKO11} & $\tO{n^{1-2/p}\eps^{-2-6/p}}$ & Arbitrary\\
\cite{BhuvanagiriGKS06} & $\tO{n^{1-2/p}\eps^{-2-4/p}}$ & Arbitrary\\
\cite{Ganguly11, GangulyW18} & $\tO{n^{1-2/p}\eps^{-2}}$ & Arbitrary \\
This work & $\tO{n^{1-2/p}\eps^{-4/p}}$ & Random or Two-Pass Arbitrary\\\hline
\cite{AlonMS99, Woodruff04} & $\Omega\left(n^{1-5/p}+\eps^{-2}\right)$ & Arbitrary \\
\cite{ChakrabartiKS03} & $\Omega\left(n^{1-2/p}\eps^{-2/p}\right)$ & Arbitrary \\
\cite{WoodruffZ12} & $\Omega\left(n^{1-2/p}\eps^{-4/p}/\log^{\O{1}}n\right)$ & Arbitrary, $\O{1}$-Passes \\
\cite{Ganguly12} & $\Omega\left(n^{1-2/p}\eps^{-2}/\log n\right)$ & Arbitrary \\
\cite{AndoniMOP08,ChakrabartiC016} & $\Omega(n^{1-2.5/p}/\log n + \varepsilon^{-2})$ & Random \\
This work & $\Omega\left(n^{1-2/p}\eps^{-2}\right)$ & Arbitrary \\\hline
\end{tabular}
\end{center}
\caption{Summary of recent work for large frequency moment estimation}
\figlab{fig:history}
\end{figure*}

There are also major gaps in our understanding for $F_p$-moment estimation for $p > 2$ in random order streams. 
The best known\footnote{After discussion with the authors, there appears to be an error in \cite{guha2009revisiting}, which claims a stronger lower bound.} lower bound is
$\Omega(n^{1-2.5/p}/\log n + \varepsilon^{-2})$ \cite{AndoniMOP08,ChakrabartiC016}, while no upper bounds that do better in random order streams than in arbitrary insertion streams are known. 

\subsection{Our Results}
In this paper, we show a separation not only between random-order and arbitrary insertion-only streams for the $F_p$ moment estimation problem, but also one-pass and multi-pass streams.  
We first show improved bounds for the $F_p$ moment estimation problem for every $p>2$ in the random-order insertion-only streaming model. 
\begin{restatable}{theorem}{thmmainlarge}
\thmlab{thm:main:large}
For $p>2$, there exists an algorithm that outputs a $(1+\eps)$-approximation to the $F_p$ moment of a random-order insertion-only stream with probability at least $\frac{2}{3}$, while using total space (in bits) $\tO{\frac{1}{\eps^{4/p}}\cdot n^{1-2/p}}$.
\end{restatable}
\thmref{thm:main:large} utilizes the random-order model to improve the algorithm on arbitrary-order insertion-only streams using $\tO{\frac{1}{\eps^2}\cdot n^{1-2/p}}$ space~\cite{Ganguly11,GangulyW18}, in terms of the dependence on $\frac{1}{\eps}$. 
We then give an algorithm that uses roughly the same bounds even for the two-pass streaming model, even if the order of the updates is adversarial.  
\begin{theorem}
\thmlab{thm:twopass:meta}
For $p>2$, there exists a two-pass streaming algorithm that outputs a $(1+\eps)$-approximation to the $F_p$ moment with probability at least $\frac{2}{3}$.
\begin{itemize}
\item
If the stream is insertion-only, the algorithm uses $\tO{\frac{1}{\eps^{4/p}}\cdot n^{1-2/p}}$ bits of space (see \thmref{thm:twopass:insertion}). 
\item
If the steam has turnstile updates, the algorithm uses $\tO{\frac{1}{\eps^{4/p}}\cdot n^{1-2/p}}$ bits of space (see \thmref{thm:twopass:turnstile}). 
\end{itemize}
\end{theorem}
\thmref{thm:twopass:meta} is the first algorithm to match the lower bound of $\tilde{\Omega}\left(\frac{1}{\eps^{4/p}}\cdot n^{1-2/p}\right)$ for multi-pass frequency moment estimation~\cite{WoodruffZ12} up to $\log n$ factors. 

By contrast, we give a lower bound for $F_p$ estimation in the one-pass insertion-only streaming model when the order of the updates is adversarial. 
\begin{restatable}{theorem}{thmlblarge}
\thmlab{thm:lb:large}
For any constant $p>2$ and parameter $\eps=\Omega\left(\frac{1}{n^{1/p}}\right)$, any one-pass insertion-only streaming algorithm that outputs a $(1+\eps)$-approximation to the $F_p$ moment of an underlying frequency vector with probability at least $\frac{9}{10}$ requires $\Omega\left(\frac{1}{\eps^2}\cdot n^{1-2/p}\right)$ bits of space. 
\end{restatable}
\thmref{thm:lb:large} improves the lower bound of $\Omega\left(\frac{1}{\eps^2 \log n}\cdot n^{1-2/p}\right)$ for general insertion-only streams for $p>2$ by \cite{Ganguly11}. 
Together, \thmref{thm:main:large} and \thmref{thm:lb:large} for small enough $\varepsilon > 0$ show a somewhat surprising result that random-order streams are strictly easier than arbitrary insertion-only streams. 
Similarly, \thmref{thm:twopass:meta} and \thmref{thm:lb:large} together show that multiple passes are strictly easier than a single pass on arbitrary insertion-only streams. 


\subsection{Our Techniques}
\seclab{sec:techniques}
We first describe our random-order insertion-only algorithm. 
At a high level, our algorithm interleaves the recent heavy-hitters algorithm of~\cite{BravermanGW20} with a subsampling procedure to estimate the contributions of various level sets toward the frequency moment. 

\textbf{Approximate frequencies of heavy-hitters.} 
The heavy-hitters algorithm of~\cite{BravermanGW20} partitions the updates of a random-order stream into blocks of updates. 
It then randomly chooses a number of coordinates from the universe $[n]$ to test in each block. 
The heavy-hitters will pass the test and subsequently be tracked across a number of following blocks before estimates for their frequencies are output. 
Due to the uniformity properties of the random-order stream, the heavy-hitters are sufficiently ``spread out'', so the algorithm crucially outputs a $(1+\eps)$-approximation to the frequency of each heavy-hitter; algorithms with the same guarantee and space complexity for arbitrary-order streams do not exist~\cite{Ganguly12}. 

\textbf{Using level sets to estimate frequency moments.} 
Given the subroutine for finding $(1+\eps)$-approximations to the frequencies of heavy-hitters, we now build upon a standard subsampling approach to estimate the frequency moment. 
Informally, we conceptually define the level sets so that level set $\Lambda_i$ contains the coordinates $k\in[n]$ such that $f_k^p\in\left[\frac{F_p}{2^i},\frac{2F_p}{2^i}\right]$. 
Since the level sets partition the universe, it is easy to see that if we define the contribution $C_i$ of a level set $\Lambda_i$ by the sum of the contributions of all their coordinates, $C_i:=\sum_{k\in \Lambda_i}f_k^p$, then $F_p$ is just the sum of all the contributions of the level sets, $F_p=\sum_i C_i$. 

\cite{IndykW05} showed that the contributions of each ``significant'' level set can be estimated by subsampling at exponentially smaller rates and considering the approximate frequencies of the heavy-hitters in each of the subsamples. 
For example, a single item with contribution $F_p$ will be detected at the top level, while $n$ items with contribution $1$ will be detected at a subsampling level where there are roughly $\Theta\left(\frac{1}{\eps^p}\right)$ survivors in expectation. 
Crucially, $(1+\eps$)-approximations to the contribution of the surviving heavy-hitters in each subsampling level can then be rescaled by the sampling rate to obtain ``good'' approximations to the contributions of each significant level set; these very good estimates are not available in standard subsampling schemes for arbitrary order streams. 

We adapt this approach for $p>2$ to obtain \thmref{thm:main:large}. 
The first observation is that an item $i$ with $f_i^p\ge\eps^2F_p$ should be identified to control the variance but also satisfies $f_i^2\ge\eps^{4/p}/n^{1-2/p}\cdot F_2$, so we can identify these items using the heavy-hitter algorithm of \cite{BravermanGW20} with the corresponding threshold; this induces the overall $n^{1-2/p}/\eps^{4/p}$ dependency. 
Moreover as we subsample, the space of the universe decreases in expectation from $n$ to $n/2$ to $n/4$ and so forth. 
Thus determining the $L_p$-heavy hitters at lower subsampling rates can be done using significantly less space. 
We can take advantage of this by requiring that our heavy-hitter algorithm aggressively seeks heavy-hitters with lower thresholds at lower subsampling rates. 
We can thus achieve a geometric series and avoid an additional $\O{\log n}$ factor in our space. 


\textbf{From random-order to two-pass arbitrary order.}
As stated, our algorithm necessitates the random-order model so that the heavy-hitter subroutine can output $(1+\eps)$-approximations to the frequencies of heavy-hitters across different subsampling levels; these approximations are then used to obtain $(1+\eps)$-approximations to the contributions of each level set. 
The state-of-the-art heavy-hitter algorithms in insertion-only~\cite{BravermanCINWW17} or turnstile~\cite{CharikarCF04} streams with arbitrary arrival order do not give a $(1+\eps)$-approximation to the frequency of each heavy-hitter while still using space dependency $\frac{1}{\eps^2}$. 
Fortunately, our approach can be remedied in two-passes over the data stream by first using a pass to identify each of the heavy-hitters across the different subsampling levels and then using the second pass to exactly count their frequencies; note that since $n^{1-2/p}/\eps^{4/p}\ge\frac{1}{\eps^2}$ for $n\ge\frac{1}{\eps^2}$, we are still permitted space to track the frequencies of $\min\left(n,\frac{1}{\eps^2}\right)$ items. 
We can then obtain $(1+\eps)$-approximations to the contributions of each significant level set, thus obtaining a $(1+\eps)$-approximation to the $F_p$ frequency moment. 

\textbf{Source of the separation.} 
In summary, our improved upper bounds in both the random-order and multi-pass models exploit each model to obtain $(1+\eps)$-approximate frequencies of the heavy-hitters in each subsampling level. 
Due to the uniformity of heavy-hitters across the stream in the random-order model, we are able to obtain $(1+\eps)$-approximate frequencies by simply tracking the frequency of the heavy-hitters within a small block of the stream and then scaling by the entire length of the stream. 
For multi-pass models, we are able to identify heavy-hitters in the first pass and exactly track their frequencies in the second pass. 
By contrast, obtaining $(1+\eps)$-approximate frequencies in adversarially ordered insertion-only streams with the same space guarantees in a single pass cannot be done~\cite{Ganguly12}. 

\textbf{Lower bound.} 
To prove our improved lower bound, we first recall the standard approach for showing $F_p$ moment estimation lower bounds in insertion-only streams for $p>2$ that uses the multiplayer set disjointness problem, e.g.,~\cite{Bar-YossefJKS04, Ganguly12}. 
In this problem, $t$ players have binary vectors of length $n$ and the promise is that the largest coordinate in the sum of all vectors is either (1) at most $1$ or (2) exactly $t$. 
For $t=\eps^{1/p}n^{1/p}$, the $F_p$ frequency moment of the sum of the binary vectors differs by a factor of $(1+\eps)$ between the two cases, so that a $(1+\eps)$-approximation to $F_p$ solves the multiplayer set disjointness problem. 
The total communication complexity of the multiplayer set disjointness problem is $\Omega\left(\frac{n}{t}\right)$ so one of the $t$ players communicates $\Omega\left(\frac{n}{t^2}\right)$ bits, and thus a lower bound of $\Omega\left(\frac{1}{\eps^{2/p}}\cdot n^{1-2/p}\right)$ follows.

To improve the $\eps$ dependency in the lower bound to $\frac{1}{\eps^2}$, we define the $(t,\eps,n)$-player set disjointness estimation problem so that there are $t+1$ players $P_1,\ldots,P_{t+1}$ with private coins in the standard blackboard model. 
The first $t$ players each receive a vector $\v_s\in\{0,1\}^n$ for $s\in[t]$, while player $P_{t+1}$ receives an index $j\in[n]$ and a bit $c\in\{0,1\}$. 
For $\u=\sum_{s\in[t]}\v_s$, the inputs are promised to satisfy $u_i\le 1$ for each $i\neq j$ and either $u_j=1$ or $u_j=t$, similar to the multiparty set disjointness problem.  
With probability at least $\frac{9}{10}$, $P_{t+1}$ must differentiate between the three possible input cases (1) $u_j+\frac{ct}{\eps}\le t$, (2) $u_j+\frac{ct}{\eps}\in\left\{\frac{t}{\eps},\frac{t}{\eps}+1\right\}$, or (3) $u_j+\frac{ct}{\eps}=(1+\eps)\frac{t}{\eps}$, where $\eps\in(0,1)$. 
We call coordinate $j\in[n]$ the \emph{spike} location and show that the $(t,\eps,n)$-player set disjointness estimation problem requires $\Omega\left(\frac{n}{t}\right)$ total communication by using a direct sum embedding to decompose the conditional information complexity into a sum of $n$ single coordinate problems. 
The intuition is that the first $t$ players do not know the spike location, so they must effectively solve the problem on each coordinate. 
We then bound the conditional information complexity of each single coordinate problem by the Hellinger distances between inputs for which the outputs differ. 
We thus apply the same reduction and set $t=\Theta\left(\frac{1}{\eps}\cdot n^{1/p}\right)$ to obtain the desired lower bound. 

We remark that the $(t,\eps,n)$-player set disjointness estimation problem can be seen as a generalization of the augmented $L_\infty$ promise problem introduced by \cite{LiW13}. 
In the augmented $L_\infty$ promise problem, there are only three players, but each coordinate in the first two players' input vectors can be as large as $\eps k$ for some fixed parameter $k$. 
\cite{LiW13} used the augmented $L_\infty$ promise problem to give a lower bound of $\Omega\left(\frac{\log n}{\eps^2}\cdot n^{1-2/p}\right)$ for $F_p$ moment estimation on turnstile streams. 
However, since their reduction crucially allows players to use turnstile updates, we cannot adapt their techniques to obtain a reduction for insertion-only streams. 

\subsection{Preliminaries}
For a positive integer $n$, we use the notation $[n]$ to denote the set of integers $\{1,2,\ldots,n\}$. 
For a frequency vector $f$ with dimension $n$ and $p\ge 2$, we define the $F_p$ moment function by $F_p(f):=\sum_{k\in[n]}|f_k|^p$ and the $L_p$ norm of $f$ by $L_p(f)=(F_p(f))^{1/p}$. 
When $f$ is defined through the updates of a insertion-only data stream, we simply use $F_p$ to denote $F_p(f)$. 
Thus for a subset $I\subseteq[n]$, the notation $F_p(I)$ is understood to mean $\sum_{k\in I}f_k^p$. 
We also use the notation $\|f\|_p$ to denote the $L_p$ norm of $f$. 
The $F_p$ moment estimation problem and the norm estimation problem are often used interchangeably, as a $(1+\eps)$-approximation algorithm to one of these problems can be modified to output a $(1+\eps)$-approximation to the other using a rescaling of $\eps$, for constant $p > 0$. 

We use $\log$ and $\ln$ to denote the base two logarithm and natural logarithms, respectively. 
We use $\poly(n)$ to denote a fixed constant degree polynomial in $n$ and $\frac{1}{\poly(n)}$ to denote some arbitrary degree polynomial in $n$ corresponding to the choice of constants in the algorithms. 
We use $\polylog(n)$ to denote polylogarithmic factors of $n$. 

The $L_p$-heavy hitters problem is to output all coordinates $i$ such that $f_i\ge\eps\cdot L_p$; such a coordinate is called a \emph{heavy-hitter}. 
The problem allows coordinates $j$ with $f_j\le\eps\cdot L_p$ to be output, provided that $f_j\ge\frac{\eps}{2}\cdot L_p$. 
Moreover, each coordinate output by the algorithm must also have a frequency estimation with additive error at most $\frac{\eps}{2}\cdot L_p$. 

\section{$F_p$ Estimation for $p>2$ in Random-Order Streams}
\seclab{sec:random}
In this section, we give our $F_p$ estimation algorithm for $p>2$. 
We first introduce an algorithm $\counthh$ on random-order insertion-only streams that outputs an approximate frequency for each $L_2$ heavy-hitter. 
\begin{restatable}{theorem}{thmcounthh}[Theorem 28 in~\cite{BravermanGW20}]
\thmlab{thm:count:hh}
There exists a one-pass algorithm $\counthh$ on random-order insertion-only streams that outputs a list $H$ of ordered pairs $(j,\widehat{f_j})$ such that $\widehat{f_j}=(1\pm2\eps)f_j$, for each $j\in H$, where $f$ is the underlying frequency vector. 
Moreover, we have that $j\in H$ for each $j\in[n]$ with $f_j^2\ge\eps^2\cdot F_2$ and $j\not\in H$ for each $j$ with $f_j^2\le\frac{\eps^2}{2}\cdot F_2$. 
The algorithm uses $\O{\frac{1}{\eps^2}\left(\log^2\frac{1}{\eps}+\log^2\log m+\log n\right)+\log m}$ bits of space and succeeds with probability at least $38/39$.
\end{restatable}
A standard heavy-hitter algorithm such as $\countsketch$~\cite{CharikarCF04} or $\bptree$~\cite{BravermanCINWW17} outputs each $j\in[n]$ with $f_j^2\ge\frac{\eps^2}{2}\cdot F_2$ but only finds a constant-factor approximation to the frequency of each heavy-hitter. 
By comparison, $\counthh$ crucially uses the random-order stream to output a $(1+\eps)$-approximation to the ferquency of each heavy-hitter; we defer the full details of the algorithm to \appref{sec:counthh}. 
We use $\counthh$ in a subsampling approach to approximate the contributions of each of the level sets, defined as follows:
\begin{definition}[Level sets and contribution]
\deflab{def:level:sets}
Given $\widehat{F_p}$ such that $\widehat{F_p}\le F_p\le1.01\cdot\widehat{F_p}$ and a uniformly random $\zeta\in[1,2]$, we define the \emph{level set} $\Lambda_i$ for each $i\in[\log 4n]$ so that
\[\Lambda_i:=\left\{k\,|\, f_k^p\in\left[\frac{\zeta\cdot\widehat{F_p}}{2^i},\frac{2\zeta\cdot\widehat{F_p}}{2^i}\right]\right\}.\]
Then we define the \emph{contribution} $C_i$ of level set $\Lambda_i$ to be $C_i:=\sum_{k\in \Lambda_i}f_k^p$. 
We define the \emph{fractional contribution} $\phi_i$ of level set $\Lambda_i$ to be the ratio $\phi_i:=\frac{C_i}{\zeta\widehat{F_p}}$, so that $\phi_i\in[0,1.01]$. 

For a stream of length $m=\poly(n)$, let $\alpha$ be an integer such that $2^\alpha>m^p$. 
We say a level set $\Lambda_i$ is \emph{significant} if its fractional contribution $\phi_i$ is at least $\frac{\eps}{2\alpha\log n}$. 
Otherwise, we say the level set is \emph{insignificant}. 
\end{definition}
Observe that it suffices to obtain a multiplicative $\left(1+\frac{\eps}{2}\right)$-approximation to the contribution of each significant level set and estimate the contributions of the insignificant level sets to be zero, since there are at most $\alpha\log n$ level sets and thus the total additive error from the insignificant level sets is at most $\frac{\eps}{2\alpha\log n}\cdot F_p\cdot \alpha\log n=\frac{\eps}{2}\cdot F_p$. 

To estimate the contribution of each level set, we use a combination of $\counthh$ and subsampling to approximate the frequencies of a number of heavy-hitters. 
The main idea is that subsampling induces a separate substream with a frequency vector with a smaller $F_p$. 
Thus the items in the level sets $\Lambda_i$ with small $i$ will be heavy-hitters of the frequency vector while the items in the level sets $\Lambda_i$ with larger $i$ will be heavy-hitters of the substreams with smaller sampling rates (if the items are subsampled), due to the lower $F_p$ of the substreams. 
We can then use a $(1+\eps)$-approximation to the frequency of each sampled heavy-hitter to estimate a $(1+\eps)$-approximation to the contribution of the level set, due to uniformity properties of random-order streams. 
Note that this is where we crucially use $\counthh$ over other possible heavy-hitter algorithms, due to its advantage of providing $(1+\eps)$-approximate frequencies in the random-order model. 
We describe our algorithm for $F_p$ estimation with $p>2$ for random-order insertion-only streams in \algref{alg:estimator}. 
\begin{algorithm}[!htb]
\caption{$F_p$ Estimation in the Random-Order Insertion-Only Model, $p>2$}
\alglab{alg:estimator}
\begin{algorithmic}[1]
\Require{Accuracy parameter $\eps\in(0,1)$, $\widehat{F_p}\le F_p\le1.01\cdot\widehat{F_p}$}
\Ensure{$(1+\eps)$-approximation to $F_p$.}
\State{Let $\zeta\in[1,2]$ be chosen uniformly at random and $\alpha$ be a positive integer such that $2^\alpha>m^p$.}
\State{Let $\eta>\sum_{j\ge0}2^{-j(p/16-1/8)}$ be a sufficiently large constant.}
\State{$\gamma\gets 2^{11}$, $n_i\gets\min\left(\left(\frac{16\alpha p\log n}{\eps^{1-2/p}}\right)^{(2p)/(p-2)},\frac{10\gamma n}{2^i}\right)$, $\eps_i=\frac{\eps}{16\eta\cdot 2^{i(p/16-1/8)}\log\frac{1}{\eps^2}}$}
\State{Let $I^r_i$ be a (nested) subset of $[n]$ subsampled at rate $p_i:=\min(1,2^{-i}\gamma)$.}
\State{Let $H^r_i$ be the output of $\counthh$ with threshold parameter $\frac{(\eps_i)^{2/p}}{80\gamma}\cdot\left(\frac{1}{n_i}\right)^{1/2-1/p}$ on the substream induced by $I^r_i$.}
\For{$i\in[\alpha\log n]$, $r\le\O{\log\log n}$}
\State{$\ell_i:=\min\{k:2^k>2^i\cdot\eps_k^2\}$}
\State{Let $S^r_i$ be the set of ordered pairs $(j,\widehat{f_j})$ in $H^r_{\ell_i}$ with $\left(\widehat{f_j}\right)^p\in\left[\frac{\zeta\widehat{F_p}}{2^i},\frac{2\zeta\widehat{F_p}}{2^i}\right]$.}
\State{$\widehat{C_i}\gets\median_r\frac{1}{p_{\ell_i}}\cdot\left(\sum_{(j,\widehat{f_j})\in S^r_i}\left(\widehat{f_j}\right)^p\right)$}

\EndFor
\State{\Return $\tilde{F_p}:=\sum_i\widehat{C_i}$}
\end{algorithmic}
\end{algorithm}

\begin{remark}
\remlab{rem:X:upper}
We remark that \algref{alg:estimator} is written requiring an input $\widehat{F_p}$ such that $\widehat{F_p}\le F_p\le1.01\cdot\widehat{F_p}$ in order to correctly index the level sets defined by \defref{def:level:sets}. 
\algref{alg:estimator} can be rewritten by setting $X=(F_1)^p$ to be an upper bound on $F_p$ and redefining the level sets in \defref{def:level:sets} so that $\Lambda_i:=\left\{k\,|\, f_k^p\in\left[\frac{\zeta\cdot X}{2^i},\frac{2\zeta\cdot X}{2^i}\right]\right\}$. 
Then we again have $F_p=\sum_i C_i$ across the contributions of all the level sets. 
\end{remark}

We first show that there exists a subsampling rate such that items in each level set will be reported as a heavy-hitter if they are successfully subsampled. 
Moreover, each item reported as a heavy-hitter will also be reported with an ``accurate'' estimate of its frequency.
\begin{restatable}{lemma}{lemsinglesamplelevel}
\lemlab{lem:single:sample:level}
Let $\eps\in(0,1)$, $\Lambda_i$ be a fixed level set, and let $\ell:=\ell_i:=\min\{k:2^k>2^i\cdot\eps_k^2\}$. 
For a fixed $r$, let $\mathcal{E}_1$ be the event that $|I^r_\ell|\le\frac{10\gamma n}{2^\ell}$ and let $\mathcal{E}_2$ be the event that $F_p(I^r_\ell)\le\frac{10\gamma F_p}{2^\ell}$. 
Conditioned on $\mathcal{E}_1$ and $\mathcal{E}_2$, there exists a $(j,\widehat{f_j})$ in $S^r_i$ for each $j\in \Lambda_i\cap I^r_\ell$ such that with probability at least $\frac{7}{8}$,
\[\left(1-\frac{\eps}{8\alpha\log n}\right)\cdot f_j^p\le\left(\widehat{f_j}\right)^p\le\left(1+\frac{\eps}{8\alpha\log n}\right)\cdot f_j^p.\]
\end{restatable}
\begin{proof}
Recall that \algref{alg:estimator} considers $H^r_\ell$ across $r\le\O{\log\log n}$ independent subsamples of $[n]$, to construct $\widehat{C_i}$ by subsampling at rate $p_\ell:=\min(1,2^{-\ell}\gamma)$. 

Suppose $2^i\cdot\eps^2<\gamma$, so that $\ell:=\ell_i=\min\{k:2^k>2^i\cdot\eps_k^2\}\le\log\gamma$ since $\eps_k\le\eps$ for all $k$. 
Then $p_\ell=\min(1,2^{-\ell}\gamma)=1$ and all items are subsampled, i.e., $H^r_\ell=[n]$. 
By \defref{def:level:sets} of the level sets, each item $j\in \Lambda_i$ satisfies $f_j^p\in\left[\frac{\zeta\cdot\widehat{F_p}}{2^i},\frac{2\zeta\cdot\widehat{F_p}}{2^i}\right]$. 
Since $\eps^2\ge\frac{1}{2^i}$, then we have $f_j^p\ge\eps^2\cdot\widehat{F_p}$. 
We also have $\widehat{F_p}\le F_p\le1.01\cdot\widehat{F_p}$ and $n^{1/2-1/p}\cdot F_p^{1/p}\ge F_2^{1/2}$. 
Thus, each item in $j\in \Lambda_i$ satisfies 
\[f_j\ge\frac{\eps^{2/p}}{1.01^{1/p}}\cdot F_p^{1/p}\ge\frac{\eps^{2/p}}{1.01^{1/p}n^{1/2-1/p}}\cdot F_2^{1/2},\]
so that $f_j^2\ge\frac{\eps^{4/p}}{1.01n^{1-2/p}}F_2$. 
Then by \thmref{thm:count:hh}, each item $j\in \Lambda_i$ corresponds to some estimate $\left(\widehat{f_j}\right)^2$ reported by $H^r_{\ell}$ output by $\counthh$ with threshold $\frac{(\eps_\ell)^{2/p}}{80\gamma}$, as $\eps_\ell\le\eps$ and $n_\ell=n$. 
Hence, $\widehat{f_j^p}$ is a $\left(1+\frac{\eps}{8\alpha\log n}\right)$-approximation to $f_j^p$. 
Furthermore, $\counthh$ rounds the estimate of the frequency of each heavy-hitter to the nearest power of $\left(1+\frac{\eps}{16p\alpha\log n}\right)$. 
Hence, $\widehat{f_j^p}$ is a $\left(1+\frac{\eps}{8\alpha\log n}\right)$-approximation to $f_j^p$. 

Now suppose $2^i\cdot\eps^2\ge\gamma$, so that $\ell:=\ell_i=\min\{k:2^k>2^i\cdot\eps_k^2\}\ge\log\gamma$ and $p_\ell:=\min(1,2^{-\ell}\gamma)\ge\frac{\gamma}{2\cdot 2^i\eps_{\ell}^2}$. 
Again by \defref{def:level:sets} of the level sets, each item $j\in \Lambda_i$ satisfies $f_j^p\in\left[\frac{\zeta\cdot\widehat{F_p}}{2^i},\frac{2\zeta\cdot\widehat{F_p}}{2^i}\right]$. 
Now if $j\notin I^r_{\ell}$, then it will not be reported by $\counthh$. 
Thus we assume that $j\in I^r_{\ell}$. 

Conditioning on the event $\mathcal{E}_2$, we have that 
\[F_p(I^r_\ell)\le\frac{10\gamma F_p}{2^\ell}\le\frac{20\gamma F_p}{2^i\eps^2}.\]
Since $\widehat{F_p}\le F_p\le1.01\cdot\widehat{F_p}$, then we have that $f_j^p\ge\frac{\eps^2}{80\gamma}\cdot F_p(I^r_\ell)$ and thus
\[f_j\ge\frac{\eps^{2/p}}{(80\gamma)^{1/p}}\cdot F_p^{1/p}(I^r_\ell).\]
Instead of applying the inequality $n^{1/2-1/p}\cdot F_p^{1/p}\ge F_2^{1/2}$, we note that conditioning on the event $\mathcal{E}_1$, we have that $|I^r_\ell|\le n_\ell=\frac{10\gamma n}{2^\ell}$. 
Hence, the frequency vector defined by the substream $I^r_\ell$ potentially has much smaller support size than $n$. 
Thus we have
\[f_j\ge\frac{\eps^{2/p}}{(80\gamma)^{1/p}}\cdot\left(\frac{1}{n_\ell}\right)^{1/2-1/p}\cdot F_2^{1/2}(I^r_\ell).\]
By \thmref{thm:count:hh}, each item $j\in \Lambda_i\cap I^r_\ell$ corresponds to some estimate $\widehat{f_j}$ reported by $H^r_{\ell}$ output by $\counthh$ with threshold $\frac{(\eps_{\ell})^{2/p}}{80\gamma}\cdot\left(\frac{1}{n_\ell}\right)^{1/2-1/p}$, since $\eps_{\ell}\le\eps$. 
Moreover, the estimate of the frequency of each heavy-hitter reported by $\counthh$ is within a factor of $\left(1+\frac{\eps}{16p\alpha\log n}\right)$ of the true frequency, since $n_i\ge\left(\frac{16\alpha p\log n}{\eps^{1-2/p}}\right)^{(2p)/(p-2)}$ implies that the threshold is at most $\frac{\eps}{16p\alpha\log n}$.  
Hence for sufficiently small $\eps$, $\left(\widehat{f_j}\right)^p$ is a $\left(1+\frac{\eps}{8\alpha\log n}\right)$-approximation to $f_j^p$. 
\end{proof}
In summary, an item $k\in \Lambda_i$ may not always be sampled by $I^r_\ell$, but $\counthh$ outputs a quantity $\widehat{f_k}$ such that $\left(\widehat{f_k}\right)^p$ is a $\left(1+\frac{\eps}{8\alpha\log n}\right)$-approximation to $f_k^p$ if $k\in I^r_\ell$. 

These approximations allow us to recover a $(1+\eps)$-approximation to the $F_p$ moment through \lemref{lem:correctness:largep}, which is our main correctness statement and which we now show. 
\lemref{lem:correctness:largep} claims that the output $\tilde{F_p}$ of our algorithm serves as a $(1+\eps)$-approximation to the $F_p$ moment of the underlying frequency vector. 
The proof of \lemref{lem:correctness:largep} first considers an idealized process and shows that we obtain ``good'' approximations to the contributions of each significant level set. 
Since the contributions of the insignificant level sets can be ignored, we thereby obtain a $(1+\eps)$-approximation to the $F_p$ moment. 
We then show that when the process is not idealized, the estimate $\tilde{F_p}$ only occurs a small error and thus still guarantees a $(1+\eps)$-approximation to the $F_p$ moment.
\begin{restatable}{lemma}{lemcorrectnesslargep}
\lemlab{lem:correctness:largep}
With probability at least $\frac{2}{3}$, we have that $|\tilde{F_p}-F_p|\le\eps\cdot F_p$.
\end{restatable}
\begin{proof}
Let $\Lambda_i$ be a fixed level set and let $\ell:=\ell_i:=\min\{k:2^k>2^i\cdot\eps_k^2\}$. 
Let $k\in \Lambda_i\cap S^r_\ell$, so that $f_k^p\in\left[\frac{\zeta\cdot\widehat{F_p}}{2^i},\frac{2\zeta\cdot\widehat{F_p}}{2^i}\right]$ and $k$ is subsampled at the level in which its estimated frequency $\left(\widehat{f_k}\right)^p$ is used to estimate the contribution $C_i$ of level set $\Lambda_i$. 
Then \lemref{lem:single:sample:level} shows that we obtain an estimate $\left(\widehat{f_k}\right)^p$ such that
\[\left(1-\frac{\eps}{8\alpha\log n}\right)\cdot f_k^p\le\left(\widehat{f_k}\right)^p\le\left(1+\frac{\eps}{8\alpha\log n}\right)\cdot f_k^p,\]
with constant probability. 
In an idealized process, we would have 
\[\frac{\zeta\cdot\widehat{F_p}}{2^i}\le\left(\widehat{f_k}\right)^p\le\frac{2\zeta\cdot\widehat{F_p}}{2^i},\]
so that the estimate $\left(\widehat{f_k}\right)^p$ is used toward the estimation $\widehat{C_i}$ of contribution $C_i$ of level set $\Lambda_i$. 
However, this may not always be the case because the value of $f_k^p$ may be near the boundary of the interval $\left[\frac{\zeta\cdot\widehat{F_p}}{2^i},\frac{2\zeta\cdot\widehat{F_p}}{2^i}\right]$ and the value of the estimate $\left(\widehat{f_k}\right)^p$ may lie outside of the interval, so that the estimate $\left(\widehat{f_k}\right)^p$ is used toward the estimation of some other level set $\Lambda_{i'}$. 
We first analyze an idealized process, so that $\left(\widehat{f_k}\right)^p$ is correctly classified for all $k$ across all level sets, and show that the output is a $(1+\O{\eps})$-approximation to $F_p$. 
We then argue that because we randomize the boundaries of each interval due to the selection of $\zeta$, the overall guarantee is only slightly worsened but remains a $(1+\eps)$-approximation to $F_p$. 

\textbf{Idealized process.} 
We first show that for an idealized process where $\left(\widehat{f_k}\right)^p$ is correctly classified for all $k$ across all level sets, then for a fixed level set $i$, we have $|\widehat{C_i}-C_i|\le\frac{\eps_{\ell}}{2}\cdot F_p$ with probability at least $1-\frac{1}{\polylog(n)}$. 
Let $\mathcal{E}_1$ be the event that $|I^r_\ell|\le\frac{10\gamma n}{2^\ell}$ and let $\mathcal{E}_2$ be the event that $F_p(I^r_\ell)\le\frac{10\gamma F_p}{2^\ell}$.
\lemref{lem:single:sample:level} shows that conditioned on $\mathcal{E}_1$ and $\mathcal{E}_2$, then $\counthh$ outputs a $\left(1+\frac{\eps}{8\alpha\log n}\right)$-approximation to $f_k^p$ if $k\in I^r_\ell$. 
We thus analyze the approximation $\widehat{C^r_i}$ to $C_i$ for a given set of subsamples, where we define
\[\widehat{C^r_i}:=\frac{1}{p_{\ell_i}}\cdot\sum_{(k,\widehat{f_k})\in S^r_i}\left(\widehat{f_k}\right)^p,\]
so that $\widehat{C_i}=\median_r\widehat{C^r_i}$. 
Conditioned on $\mathcal{E}_1$ and $\mathcal{E}_2$, we note that $\widehat{C^r_i}$ is a $\left(1+\frac{\eps}{8\alpha\log n}\right)$-approximation to
\[D^r_i:=\frac{1}{p_{\ell_i}}\cdot\sum_{k\in S^r_i\cap \Lambda_i}f_k^p.\]
We thus analyze the expectation and variance of $D^r_i$. 

We first analyze the expectation of $D^r_i$. 
Note that 
\[\Ex{D^r_i}=\frac{1}{p_{\ell}}\cdot\sum_{k\in \Lambda_i} p_{\ell}\cdot f_k^p=C_i.\]
We next analyze the variance of $D^r_i$, which results from whether items $k\in \Lambda_i$ are sampled by $I^r_\ell$. 
Since $p_\ell\ge\frac{\gamma}{2\cdot 2^i\eps_{\ell}^2}$, we have
\[\Var(D^r_i)=\frac{1}{p^2_{\ell}}\cdot\sum_{k\in \Lambda_i} p_{\ell}f_k^{2p}\le\sum_{k\in \Lambda_i}\frac{2\cdot 2^i\eps_{\ell}^2}{\gamma}\cdot f_k^{2p}.\]
Observe that for each $k\in \Lambda_i$, we have $f_k^{2p}\le\frac{16(F_p)^2}{2^{2i}}$ and $\frac{|\Lambda_i|}{2^i}\le\phi_i\le 1$. 
Thus for $\gamma=2^{11}$, 
\[\Var(D^r_i)\le|\Lambda_i|\cdot\frac{2\cdot 2^i\eps_{\ell}^2}{\gamma}\cdot\frac{16(F_p)^2}{2^{2i}}\le\phi_i\eps_{\ell}^2(F_p)^2\le\frac{\eps_{\ell}^2}{64}(F_p)^2.\]
Hence, by Chebyshev's inequality, we have that
\[\PPr{|D^r_i-C_i|\ge\frac{\eps_{\ell}}{2}\cdot F_p}\le\frac{1}{16}.\]
Since $C_i\le F_p$ and $\widehat{C^r_i}$ is a $\left(1+\frac{\eps}{8\alpha\log n}\right)$-approximation to $D^r_i$, then $\widehat{C^r_i}$ gives an approximation to $C_i$ with additive error at most $\left(\frac{\eps}{4\alpha\log n}+\eps_{\ell}\right)\cdot F_p$ with probability at least $\frac{15}{16}$, conditioned on the events $\mathcal{E}_1$ and $\mathcal{E}_2$, and the correctness of the subroutine $\counthh$. 
By \thmref{thm:count:hh}, the correctness of $\counthh$ occurs with probability at least $\frac{38}{39}$. 
By a union bound, we have that conditioned on $\mathcal{E}_1$ and $\mathcal{E}_2$, then 
\[\PPr{|\widehat{C^r_i}-C_i|\le\left(\frac{\eps}{4\alpha\log n}+\eps_{\ell}\right)\cdot F_p}\ge\frac{7}{8}.\]

For a fixed $r$, let $\mathcal{E}_1$ be the event that $|I^r_\ell|\le\frac{10\gamma n}{2^\ell}$ and let $\mathcal{E}_2$ be the event that $F_p(I^r_\ell)\le\frac{10\gamma F_p}{2^\ell}$. 
Recall that $I^r_\ell$ is a nested subset of $[n]$ subsampled at rate $p_\ell:=\min(1,2^{-\ell}\gamma)$. 
Thus we have $\Ex{|I^r_\ell|}\le\frac{\gamma n}{2^\ell}$, so that by Markov's inequality, we have that $\PPr{\mathcal{E}_1}\ge\frac{9}{10}$. 
Similarly, we have $\Ex{F_p(I^r_\ell)}\le\frac{\gamma F_p}{2^\ell}$, so that by Markov's inequality, we have that $\PPr{\mathcal{E}_2}\ge\frac{9}{10}$. 
Hence by a union bound, $\PPr{\mathcal{E}_1\wedge\mathcal{E}_2}\ge\frac{8}{10}$. 

By \lemref{lem:single:sample:level}, conditioned on the events $\mathcal{E}_1$ and $\mathcal{E}_2$, we have that $|\widehat{C^r_i}-C_i|\le\left(\frac{\eps}{4\alpha\log n}+\eps_{\ell}\right)\cdot F_p$ for a fixed $r$, with probability at least $\frac{7}{8}$. 
Thus by a union bound, we have that $|\widehat{C^r_i}-C_i|\le\left(\frac{\eps}{4\alpha\log n}+\eps_{\ell}\right)\cdot F_p$ for a fixed $r$, with probability at least $\frac{5}{8}$. 
Since $\widehat{C_i}=\median_r\widehat{C^r_i}$ across $r\le\O{\log\log n}$ iterations, then we have that $|\widehat{C_i}-C_i|\le\left(\frac{\eps}{4\alpha\log n}+\eps_{\ell}\right)\cdot F_p$ with probability at least $1-\frac{1}{\polylog(n)}$. 

By a union bound over the $\alpha\log n$ level sets, then with probability at least $1-\frac{1}{\polylog(n)}$, we have that $|\widehat{C_i}-C_i|\le\left(\frac{\eps}{4\alpha\log n}+\eps_{\ell_i}\right)\cdot F_p$ simultaneously for all $i\in[\alpha\log n]$, where $\ell_i:=\min\{k:2^k>2^i\cdot\eps_k^2\}$. 
We form our estimate $\tilde{F_p}$ to $F_p$ by setting $\tilde{F_p}:=\sum_i\widehat{C_i}$ and we have $F_p=\sum_i C_i$. 
Since $i\in[\alpha\log n]$, $\ell_i:=\min\{k:2^k>2^i\cdot\eps_k^2\}$, and $\eps_{\ell_i}=\frac{\eps}{16\eta\cdot 2^{\ell_i(p/16-1/8)}\log\frac{1}{\eps^2}}$, then
\begin{align*}
|\tilde{F_p}-F_p|&=\left|\sum_i\widehat{C_i}-\sum_i C_i\right|\le\sum_i|\widehat{C_i}-C_i|\le\sum_i\left(\frac{\eps}{4\alpha\log n}+\eps_{\ell_i}\right)\cdot F_p\\
&\le\sum_{i\in[\alpha\log n]}\frac{\eps}{4\alpha\log n}\cdot F_p+\sum_{i:\ell_i\le\log\frac{1}{\eps^2}}\eps_{\ell_i}\cdot F_p+\sum_{i:\ell_i>\log\frac{1}{\eps^2}}\eps_{\ell_i}\cdot F_p\\
&\le\frac{\eps}{4}\cdot F_p+\log\frac{1}{\eps^2}\cdot\frac{\eps}{16\log\frac{1}{\eps^2}}\cdot F_p+\frac{\eps}{16}\cdot F_p,
\end{align*}
where the last bound on the last term follows from $\eta>\sum_{j\ge0}2^{-j(p/16-1/8)}$. 
Thus we have that $|\tilde{F_p}-F_p|\le\frac{\eps}{2}\cdot F_p$ with probability at least $1-\frac{1}{\polylog(n)}$ in an idealized process. 

\textbf{Effects of randomized boundaries.} 
We say that for a fixed $r$, an item $k\in[n]$ is \emph{misclassified} if there exists a level set $\Lambda_i$ such that
\[\frac{\zeta\cdot F_p}{2^i}\le f_k^p\le\frac{2\zeta\cdot F_p}{2^i},\]
but for an estimate $\left(\widehat{f_k}\right)^p$ output by $\counthh$ on the set $S^r_{\ell_i}$, we have
\[\frac{\zeta\cdot F_p}{2^i}\le\left(\widehat{f_k}\right)^p\le\frac{2\zeta\cdot F_p}{2^i}.\]
Recall that by \lemref{lem:single:sample:level}, we have for any fixed value of $\zeta$ that
\[\left(1-\frac{\eps}{8\alpha\log n}\right)\cdot f_k^p\le\left(\widehat{f_k}\right)^p\le\left(1+\frac{\eps}{8\alpha\log n}\right)\cdot f_k^p.\]
Since $\zeta\in[1,2]$, then the probability that item $k\in[n]$ is misclassified is at most $\frac{\eps}{2\alpha\log n}$. 
Moreover, in the event that item $k\in \Lambda_i$ is misclassified, it can only be misclassified into either level set $\Lambda_{i+1}$ or level set $\Lambda_{i-1}$, since $\left(\widehat{f_k}\right)^p$ is a $\left(1\pm\frac{\eps}{8\alpha\log n}\right)$ multiplicative approximation to $f_k^p$. 

Thus in the event that item $k\in[n]$ is misclassified, then $\left(\widehat{f_k}\right)^p$ will be rescaled by an incorrect probability, but only by at most a factor of two. 
Hence the error in the computation of the contribution of $f_k^p$ to some level set $\Lambda_i$ is at most $2f_k^p$. 
Then in expectation across all $k\in[n]$, the error due to the misclassification is at most $2F_p\cdot\frac{\eps}{2\alpha\log n}=\frac{\eps}{\alpha\log n}\cdot F_p$. 
Hence by Markov's inequality for sufficiently large $n$, the misclassification error is at most an additive $\frac{\eps}{2}\cdot F_p$ with probability at least $\frac{3}{4}$. 
Therefore in total, we have that $|\tilde{F_p}-F_p|\le\eps\cdot F_p$ with probability at least $\frac{2}{3}$. 
\end{proof}
It remains to analyze the space complexity of the algorithm as well as remove some additional unnecessary assumptions. 
\thmmainlarge*
\begin{proof}
We first observe that we only require a $1.01$-approximation $\widehat{F_p}$ to $F_p$ as the input of \algref{alg:estimator} to create the level sets and analyze accordingly. 
However, we can instead define the level sets using any upper bound $X$ on $\widehat{F_p}$ by setting level set $\Lambda_i$ to be the indices $k\in[n]$ for which $f_k^p\in\left[\frac{\zeta X}{2^i},\frac{2\zeta X}{2^i}\right]$, rather than $\left[\frac{\zeta\widehat{F_p}}{2^i},\frac{2\zeta\widehat{F_p}}{2^i}\right]$, and the same analysis will follow (see \remref{rem:X:upper}). 
Intuitively, the fluidity of the definition of the level sets can be seen from the fact that we randomize the boundaries by going through a multiplicative $\zeta$ chosen randomly from $[1,2]$ anyway, and additional empty level sets will not change the approximation guarantee. 
Thus by \lemref{lem:correctness:largep}, there exists an algorithm that outputs a $(1+\eps)$-approximation to the $F_p$ moment. 

It remains to analyze the space complexity. 
By \thmref{thm:count:hh}, $\counthh$ with threshold $\eps$ requires $\O{\frac{1}{\eps^2}\left(\log^2\frac{1}{\eps}+\log^2\log m+\log n\right)+\log m}$ bits of space. 
\algref{alg:estimator} runs a separate instance of $\counthh$ with threshold $\frac{(\eps_i)^{2/p}}{80\gamma}\cdot\left(\frac{1}{n_i}\right)^{1/2-1/p}$ to output a set $H^r_i$ for $r\le\O{\log\log n}$, where $\gamma$ is a sufficiently large constant. 
Thus the total space for the $\counthh$ instances across all $i$ for a fixed $r$ is at most
\[C_1\log n+\left(C_1^2\log n\right)\cdot\sum_{i\in[\alpha\log n]}\frac{C_1(n_i)^{1-2/p}}{(\eps_i)^{4/p}},\]
for some positive constants $C_1,\alpha>0$. 
In particular, we have $\eps_i=\frac{\eps}{16\eta\cdot 2^{i(p/16-1/8)}\log\frac{1}{\eps^2}}$ and $n_i=\min\left(\left(\frac{16\alpha p\log n}{\eps^{1-2/p}}\right)^{(2p)/(p-2)},\frac{10\gamma n}{2^i}\right)$ for a sufficiently large constant $\eta$. 
Then the total space for the $\counthh$ instances across all $i$ for a fixed $r$ is at most
\[C_1\log n+\left(C_1\log^2 n\right)\cdot\sum_{i=1}^{\alpha\log n}\left(\frac{C_2n^{1-2/p}}{\eps^{4/p}}\cdot\frac{2^{i(1/4-1/(2p))}}{2^{i(1-2/p)}}\log^2\frac{1}{\eps}+\frac{C_2\log^2 n}{\eps^2}\right),\]
for some positive constants $C_1,C_2>0$. 
Observe that $\sum_{i=1}^\infty\frac{2^{i(1/4-1/(2p))}}{2^{i(1-2/p)}}=\sum_{i=1}^\infty\frac{1}{2^{3i(1-2/p)/4}}$ is a geometric series that is upper bounded by an absolute constant. 
Hence, the total space across all $i$ for a fixed $r$ is at most $\O{\frac{1}{\eps^{4/p}}\cdot n^{1-2/p}\log^2 n\log^2\frac{1}{\eps}+\frac{1}{\eps^2}\log^5 n}$ and the total space across all $r\le\O{\log\log n}$ is 
\[\O{\frac{1}{\eps^{4/p}}\cdot n^{1-2/p}\log^2 n\log^2\frac{1}{\eps}\log\log n+\frac{1}{\eps^2}\log^5 n\log\log n},\] 
which is $\tO{\frac{1}{\eps^{4/p}}\cdot n^{1-2/p}}$ space in total, since $n\ge\frac{1}{\eps^2}$ implies $\frac{1}{\eps^{4/p}}\cdot n^{1-2/p}\ge\frac{1}{\eps^2}$. 
\end{proof}

\section{$F_p$ Estimation for $p>2$ in Two-Pass Streams}
\seclab{sec:twopass}
In this section, we consider two-pass algorithms for $F_p$ estimation, with $p>2$. 
For turnstile streams, we require the guarantees of the well-known $\countsketch$ algorithm for finding heavy-hitters. 
\begin{theorem}
\cite{CharikarCF04}
\thmlab{thm:countsketch}
There exists an algorithm $\countsketch$ that reports all items $i\in[n]$ such that $f_i\ge\eps L_2$ and no items $j\in[n]$ such that $f_j\le\frac{\eps}{2}\cdot L_2$ in the turnstile streaming algorithm.  
The algorithm uses $\O{\frac{1}{\eps^2}\log^2 n}$ bits of space and succeeds with probability $1-\frac{1}{\poly(n)}$. 
\end{theorem}
For insertion-only streams, we use the more space-efficient $\bptree$ algorithm for finding heavy-hitters.
\begin{theorem}
\cite{BravermanCINWW17}
\thmlab{thm:bptree}
There exists an algorithm $\bptree$ that reports all items $i\in[n]$ such that $f_i\ge\eps L_2$ and no items $j\in[n]$ such that $f_j\le\frac{\eps}{2}\cdot L_2$ in the turnstile streaming algorithm.  
The algorithm uses $\O{\frac{1}{\eps^2}\log n}$ bits of space and succeeds with probability $0.99$. 
\end{theorem}

Recall that the main idea of our one-pass algorithm in the random-order insertion-only streaming model was to subsample at different rates and find $(1+\eps)$-approximations to the  frequencies of the heavy-hitters in each subsampling rate, which are then used to form estimates of the contributions of each level set and ultimately estimate of the frequency moment. 
The random-order setting crucially allowed us to obtain $(1+\eps)$-approximations to the frequencies of the heavy-hitters. 
By contrast, in an adversarial-order setting, by the time a possible heavy-hitter is detected, a significant fraction of its frequency may have already appeared in the stream if we use a heavy-hitter algorithm with space dependency $\frac{1}{\eps^2}$. 
Fortunately, in a two-pass setting, we can use the first pass to identify possible heavy-hitters and the second pass to find their frequencies. 
We give the full details in \algref{alg:multipass:estimator}, where the subroutine $\heavyhitters$ denotes the algorithm $\countsketch$ of \thmref{thm:countsketch} for turnstile streams and $\bptree$ of \thmref{thm:bptree} for insertion-only streams. 

\begin{algorithm}[!htb]
\caption{$F_p$ Estimation on Two-Pass Streams, $p>2$}
\alglab{alg:multipass:estimator}
\begin{algorithmic}[1]
\Require{Accuracy parameter $\eps\in(0,1)$, $\widehat{F_p}\le F_p\le1.01\cdot\widehat{F_p}$}
\Ensure{$(1+\eps)$-approximation to $F_p$}
\State{Let $\eta>\sum_{j\ge0}2^{-j(p/16-1/8)}$ be a sufficiently large constant.}
\State{$\gamma\gets 2^{11}$, $n_i\gets\frac{10\gamma n}{2^i}$, $\eps_i=\frac{\eps}{16\eta\cdot 2^{i(p/16-1/8)}\log\frac{1}{\eps^2}}$}
\State{Let $I^r_i$ be a (nested) subset of $[n]$ subsampled at rate $p_i:=\min(1,2^{-i}\gamma)$.}
\For{\textbf{first pass} $i\in[\alpha\log n]$, $r\le\O{\log\log n}$:}
\State{Let $H^r_i$ be the output of $\heavyhitters$ with threshold parameter $\frac{(\eps_i)^{2/p}}{80\gamma}\cdot\left(\frac{1}{n_i}\right)^{1/2-1/p}$ on the substream induced by $I^r_i$.}
\EndFor
\For{\textbf{second pass}}
\State{Track the frequency $f_k$ for any coordinate $k\in\cup H^r_i$.}
\State{Let $S^r_i$ be the items $k\in[n]$ with $f_k^p\in\left[\frac{\widehat{F_p}}{2^i},\frac{2\widehat{F_p}}{2^i}\right]$.}
\State{$\ell_i\gets\min\{k:2^k>2^i\cdot\eps_k^2\}$}
\State{$D^r_i\gets\frac{1}{p_{\ell_i}}\cdot\left(\sum_{k\in S^r_i}f_k^p\right)$}
\State{$\widehat{C_i}\gets\median_r D^r_i$}
\EndFor
\State{\Return $\tilde{F_p}:=\sum_i\widehat{C_i}$}
\end{algorithmic}
\end{algorithm}

Using $\countsketch$ as the subroutine for $\heavyhitters$ on two-pass turnstile streams and $\bptree$ as the subroutine for $\heavyhitters$ on two-pass insertion-only streams, we obtain the following guarantees for \algref{alg:multipass:estimator}. 
We first show that there exists a subsampling rate such that items in each level set will be reported as a heavy-hitter. 
The proof is similar to the proof of \lemref{lem:single:sample:level}, but we no longer require each reported item to also be reported with a $(1+\eps)$-approximation to their frequency. 
In fact, this cannot be done in a single pass; we will instead track their frequencies in the second pass. 
\begin{lemma}
\lemlab{lem:multipass:item}
Let $\eps\in(0,1)$, $\Lambda_i$ be a fixed level set, and let $\ell:=\ell_i:=\min\{k:2^k>2^i\cdot\eps_k^2\}$. 
For a fixed $r$, let $\mathcal{E}_1$ be the event that $|I^r_\ell|\le\frac{10\gamma n}{2^\ell}$ and let $\mathcal{E}_2$ be the event that $F_p(I^r_\ell)\le\frac{10\gamma F_p}{2^\ell}$. 
Conditioned on $\mathcal{E}_1$ and $\mathcal{E}_2$, then $\heavyhitters$ reports $j\in S^r_i$ for each $j\in \Lambda_i\cap I^r_\ell$ with probability at least $1-\frac{1}{\poly(n)}$. 
\end{lemma}
\begin{proof}
We consider casework on the value of $i$. 
First suppose $2^i\cdot\eps^2<\gamma$, so that $\ell:=\ell_i=\min\{k:2^k>2^i\cdot\eps_k^2\}\le\log\gamma$ since $\eps_k\le\eps$ for all $k$. 
By \defref{def:level:sets} of the level sets, each item $j\in \Lambda_i$ satisfies $f_j^p\in\left[\frac{\zeta\cdot\widehat{F_p}}{2^i},\frac{2\zeta\cdot\widehat{F_p}}{2^i}\right]$. 
We also have $\widehat{F_p}\le F_p\le1.01\cdot\widehat{F_p}$ and $n^{1/2-1/p}\cdot F_p^{1/p}\ge F_2^{1/2}$. 
Thus, each item in $j\in \Lambda_i$ satisfies 
\[f_j\ge\frac{\eps^{2/p}}{1.01^{1/p}}\cdot F_p^{1/p}\ge\frac{\eps^{2/p}}{1.01^{1/p}n^{1/2-1/p}}\cdot F_2^{1/2},\]
so that $f_j^2\ge\frac{\eps^{4/p}}{1.01n^{1/2-1/p}}F_2$. 
Then by \thmref{thm:countsketch} or \thmref{thm:bptree}, each item $j\in \Lambda_i$ is reported by $\countsketch$ or $\bptree$ with threshold $\frac{(\eps_\ell)^{2/p}}{80\gamma}$, as $\eps_\ell\le\eps$ and $n_\ell=n$. 

Otherwise, suppose $2^i\cdot\eps^2\ge\gamma$, so that $\ell:=\ell_i=\min\{k:2^k>2^i\cdot\eps_k^2\}\ge\log\gamma$ and $p_\ell:=\min(1,2^{-\ell}\gamma)\ge\frac{\gamma}{2\cdot 2^i\eps_{\ell}^2}$. 
Again by \defref{def:level:sets} of the level sets, each item $j\in \Lambda_i$ satisfies $f_j^p\in\left[\frac{\zeta\cdot\widehat{F_p}}{2^i},\frac{2\zeta\cdot\widehat{F_p}}{2^i}\right]$. 
If $j\notin I^r_{\ell}$, then $j$ certainly will not be reported by $\heavyhitters$ (regardless of whether $\heavyhitters$ is $\countsketch$ or $\bptree$). 
Hence, we assume that $j\in I^r_{\ell}$. 

Conditioning on the event $\mathcal{E}_2$,
\[F_p(I^r_\ell)\le\frac{10\gamma F_p}{2^\ell}\le\frac{20\gamma F_p}{2^i\eps^2}.\]
Given $\widehat{F_p}\le F_p\le1.01\cdot\widehat{F_p}$, then $f_j^p\ge\frac{\eps^2}{80\gamma}\cdot F_p(I^r_\ell)$. 
Therefore,
\[f_j\ge\frac{\eps^{2/p}}{(80\gamma)^{1/p}}\cdot F_p^{1/p}(I^r_\ell).\]
Rather than applying the inequality $n^{1/2-1/p}\cdot F_p^{1/p}\ge F_2^{1/2}$, we observe that conditioning on the event $\mathcal{E}_1$, it follows that $|I^r_\ell|\le n_\ell=\frac{10\gamma n}{2^\ell}$. 
Thus, the frequency vector defined by the substream $I^r_\ell$ has significantly smaller support size than $n$, which we can leverage to use a heavy-hitter algorithm with a lower threshold. 
We have
\[f_j\ge\frac{\eps^{2/p}}{(80\gamma)^{1/p}}\cdot\left(\frac{1}{n_\ell}\right)^{1/2-1/p}\cdot F_2^{1/2}(I^r_\ell).\]
Therefore by \thmref{thm:countsketch} or \thmref{thm:bptree}, each item $j\in \Lambda_i\cap I^r_\ell$ will be reported by $H^r_{\ell}$ output by $\heavyhitters$ with threshold $\frac{(\eps_{\ell})^{2/p}}{80\gamma}\cdot\left(\frac{1}{n_\ell}\right)^{1/2-1/p}$, since $\eps_{\ell}\le\eps$. 
\end{proof}
We now show our main correctness statement for our two-pass algorithms. 
\lemref{lem:multipass:correctness} proves that the output $\tilde{F_p}$ of our algorithm gives a $(1+\eps)$-approximation to the $F_p$ moment of the underlying frequency vector. 
Although the guarantees of \lemref{lem:multipass:correctness} are similar to the guarantees of \lemref{lem:correctness:largep} are similar, the proof of \lemref{lem:multipass:correctness} is much simpler. 
We show that we obtain $(1+\eps)$-approximations to the contributions of each significant level set, thus obtaining a $(1+\eps)$-approximation to the $F_p$ moment, since the contributions of the insignificant level sets can be ignored. 
Unlike \lemref{lem:correctness:largep}, we need not concern about an idealized process since the frequency of each heavy-hitter is exactly tracked in the second pass of the algorithm over the data stream, so no heavy-hitters can be accidentally misclassified into an incorrect level set. 
\begin{lemma}
\lemlab{lem:multipass:correctness}
With probability at least $\frac{2}{3}$, we have that $|\tilde{F_p}-F_p|\le\eps\cdot F_p$.
\end{lemma}
\begin{proof}
Let $\Lambda_i$ be a fixed level set and $\ell:=\min\{k:2^k>2^i\cdot\eps_k^2\}$. 
Let $\mathcal{E}_1$ be the event that $|I^r_\ell|\le\frac{10\gamma n}{2^\ell}$ and let $\mathcal{E}_2$ be the event that $F_p(I^r_\ell)\le\frac{10\gamma F_p}{2^\ell}$.
By \lemref{lem:multipass:item}, $\countsketch$ returns each $k\in \Lambda_i\cap I^r_\ell$ in $S^r_i$, conditioned on $\mathcal{E}_1$ and $\mathcal{E}_2$. 
Thus in the second pass, \algref{alg:multipass:estimator} tracks the contribution $f_k^p$ explicitly, by tracking $f_k$. 
We first analyze the approximation $\widehat{C^r_i}$ to $C_i$ for a given set of subsamples, where we define
\[D^r_i:=\frac{1}{p_{\ell}}\cdot\sum_{k\in \Lambda_i\cap I^r_\ell}f_k^p.\]
so that $\widehat{C_i}=\median_r D^r_i$. 
Conditioned on $\mathcal{E}_1$ and $\mathcal{E}_2$, we note that 
\[\Ex{D^r_i}=\frac{1}{p_{\ell}}\cdot\sum_{k\in \Lambda_i} p_{\ell}\cdot f_k^p=C_i.\]
We also have
\[\Var(D^r_i)=\frac{1}{p^2_{\ell}}\cdot\sum_{k\in \Lambda_i} p_{\ell}f_k^{2p}\le\sum_{k\in \Lambda_i}\frac{2\cdot 2^i\eps_{\ell}^2}{\gamma}\cdot f_k^{2p},\]
since $p_\ell\ge\frac{\gamma}{2\cdot 2^i\eps_{\ell}^2}$. 
Observe that for each $k\in \Lambda_i$, we have $f_k^{2p}\le\frac{16(F_p)^2}{2^{2i}}$ and $\frac{|\Lambda_i|}{2^i}\le\phi_i\le 1$. 
Thus for $\gamma=2^{11}$, 
\[\Var(D^r_i)\le|\Lambda_i|\cdot\frac{2\cdot 2^i\eps_{\ell}^2}{\gamma}\cdot\frac{16(F_p)^2}{2^{2i}}\le\phi_i\eps_{\ell}^2(F_p)^2\le\frac{\eps_{\ell}^2}{64}(F_p)^2.\]
Hence, by Chebyshev's inequality, we have that
\[\PPr{|D^r_i-C_i|\ge\frac{\eps_{\ell}}{2}\cdot F_p}\le\frac{1}{16}.\]
In summary, $D^r_i$ gives an approximation to $C_i$ with additive error at most $\frac{\eps_{\ell}}{2}\cdot F_p$ with probability at least $\frac{15}{16}$, conditioned on the events $\mathcal{E}_1$ and $\mathcal{E}_2$, and the correctness of the subroutine $\countsketch$. 
Since $\countsketch$ fails with probability at most $1-\frac{1}{\poly(n)}$ by \thmref{thm:countsketch}, then by a union bound, we have that conditioned on $\mathcal{E}_1$ and $\mathcal{E}_2$, 
\[\PPr{|D^r_i-C_i|\le\left(\frac{\eps}{4\alpha\log n}+\eps_{\ell}\right)\cdot F_p}\ge\frac{7}{8}.\]

For a fixed $r$, let $\mathcal{E}_1$ be the event that $|I^r_\ell|\le\frac{10\gamma n}{2^\ell}$ and let $\mathcal{E}_2$ be the event that $F_p(I^r_\ell)\le\frac{10\gamma F_p}{2^\ell}$. 
Recall that $I^r_\ell$ is a nested subset of $[n]$ subsampled at rate $p_\ell:=\min(1,2^{-\ell}\gamma)$, so that $\Ex{|I^r_\ell|}\le\frac{\gamma n}{2^\ell}$. 
Thus by Markov's inequality, $\PPr{\mathcal{E}_1}\ge\frac{9}{10}$. 
Similarly, $\Ex{F_p(I^r_\ell)}\le\frac{\gamma F_p}{2^\ell}$. 
Thus by Markov's inequality, $\PPr{\mathcal{E}_2}\ge\frac{9}{10}$. 
By a union bound, we first have $\PPr{\mathcal{E}_1\wedge\mathcal{E}_2}\ge\frac{8}{10}$. 
Applying another union bound, we have that $|D^r_i-C_i|\le\frac{\eps_{\ell}}{2}\cdot F_p$ for a fixed $r$, with probability at least $\frac{5}{8}$. 
Since $\widehat{C_i}=\median_r D^r_i$ across $r\le\O{\log\log n}$ iterations, then we have that $|\widehat{C_i}-C_i|\le\frac{\eps_{\ell}}{2}\cdot F_p$ with probability at least $1-\frac{1}{\polylog(n)}$. 

By a union bound over the indices $i\in[\alpha\log n]$, corresponding to the $\alpha\log n$ level sets, then with probability at least $1-\frac{1}{\polylog(n)}$, we have that $|\widehat{C_i}-C_i|\le\frac{\eps_{\ell}}{2}\cdot F_p$ simultaneously for all $i\in[\alpha\log n]$, where $\ell_i:=\min\{k:2^k>2^i\cdot\eps_k^2\}$. 
We form our estimate $\tilde{F_p}$ to $F_p$ by setting $\tilde{F_p}:=\sum_i\widehat{C_i}$ and we have $F_p=\sum_i C_i$. 
Since $i\in[\alpha\log n]$, $\ell_i:=\min\{k:2^k>2^i\cdot\eps_k^2\}$, and $\eps_{\ell_i}=\frac{\eps}{16\eta\cdot 2^{\ell_i(p/16-1/8)}\log\frac{1}{\eps^2}}$, then
\begin{align*}
|\tilde{F_p}-F_p|&=\left|\sum_i\widehat{C_i}-\sum_i C_i\right|\le\sum_i|\widehat{C_i}-C_i|\le\sum_i\left(\frac{\eps}{4\alpha\log n}+\eps_{\ell_i}\right)\cdot F_p\\
&\le\sum_{i\in[\alpha\log n]}\frac{\eps}{4\alpha\log n}\cdot F_p+\sum_{i:\ell_i\le\log\frac{1}{\eps^2}}\eps_{\ell_i}\cdot F_p+\sum_{i:\ell_i>\log\frac{1}{\eps^2}}\eps_{\ell_i}\cdot F_p\\
&\le\frac{\eps}{4}\cdot F_p+\log\frac{1}{\eps^2}\cdot\frac{\eps}{16\log\frac{1}{\eps^2}}\cdot F_p+\frac{\eps}{16}\cdot F_p,
\end{align*}
where the last bound on the last term follows from $\eta>\sum_{j\ge0}2^{-j(p/16-1/8)}$. 
Therefore with probability at least $1-\frac{1}{\polylog(n)}$, we have that $|\tilde{F_p}-F_p|\le\frac{\eps}{2}\cdot F_p$. 
\end{proof}

We now justify the full guarantees claimed by \thmref{thm:twopass:meta}. 
We first handle two passes over a turnstile stream. 
\begin{restatable}{theorem}{thmtwopassturnstile}
\thmlab{thm:twopass:turnstile}
For $p>2$, there exists a two-pass turnstile streaming algorithm that outputs a $(1+\eps)$-approximation to the $F_p$ moment with probability at least $\frac{2}{3}$, while using $\O{\frac{1}{\eps^{4/p}}\cdot n^{1-2/p}\log^2 n\log\log n}$ bits of space.
\end{restatable}
\begin{proof}
Our analysis is similar to the proof of \thmref{thm:main:large}. 
We again observe that by redefining the level sets according to any upper bound on $F_p$, we do not require a $1.01$-approximation $\widehat{F_p}$ to $F_p$ as the input of \algref{alg:multipass:estimator}. 
Thus by \lemref{lem:multipass:correctness}, there exists an algorithm that outputs a $(1+\eps)$-approximation to the $F_p$ moment on two pass turnstile streams and it remains to analyze the space complexity. 
By \thmref{thm:countsketch}, $\countsketch$ with threshold $\eps$ requires $\O{\frac{1}{\eps^2}\log^2 n}$ bits of space to output the indices of the heavy-hitters. 
\algref{alg:multipass:estimator} runs a separate instance of $\countsketch$ with threshold $\frac{(\eps_i)^{2/p}}{80\gamma}\cdot\left(\frac{1}{n_i}\right)^{1/2-1/p}$ to output a set $H^r_i$ for $r\le\O{\log\log n}$, where $\gamma$ is a sufficiently large constant. 
Thus the total space in the first pass across all indices $i$ for a fixed $r$ is at most
\[\left(C_1\log^2 n\right)\cdot\sum_{i\in[\alpha\log n]}\frac{(n_i)^{1-2/p}}{(\eps_i)^{4/p}},\]
for some positive constants $C_1,\alpha>0$. 
In particular, we have $n_i=\frac{10\gamma n}{2^i}$ and $\eps_i=\frac{\eps}{16\eta\cdot 2^{i(p/16-1/8)}\log\frac{1}{\eps^2}}$ for a sufficiently large constant $\eta$. 
Then the total space in the first pass across all indices $i$ for a fixed $r$ is at most 
\[\left(C_1\log^2 n\right)\cdot\sum_{i=1}^\infty\frac{C_2n^{1-2/p}}{\eps^{4/p}}\cdot\frac{2^{i(1/4-1/(2p))}}{2^{i(1-2/p)}},\]
for some positive constants $C_1,C_2>0$. 
Since $\sum_{i=1}^\infty\frac{2^{i(1/4-1/(2p))}}{2^{i(1-2/p)}}=\sum_{i=1}^\infty\frac{1}{2^{3i(1-2/p)/4}}$ is a geometric series that is upper bounded by a fixed constant, the total space in the first pass across all $i$ for a fixed $r$ is $\O{\frac{1}{\eps^{4/p}}\cdot n^{1-2/p}\log^2 n}$. 
Because $r\in\O{\log\log n}$, then the total space is $\O{\frac{1}{\eps^{4/p}}\cdot n^{1-2/p}\log^2 n\log\log n}$. 

In the second pass, we track the frequencies of each item reported by some instance of $\countsketch$. 
Since at most $\O{\frac{1}{\eps^{4/p}}\cdot n^{1-2/p}\log\log n}$ indices can be reported across all instances of $\countsketch$ and $\O{\log n}$ bits of space can be used to track the frequency of each reported index, then the total space for the second pass is at most $\O{\frac{1}{\eps^{4/p}}\cdot n^{1-2/p}\log n\log\log n}$. 
Thus, the total space is $\O{\frac{1}{\eps^{4/p}}\cdot n^{1-2/p}\log^2 n\log\log n}$.
\end{proof}

Finally, we note that to report at most $\O{\frac{1}{\eps^{4/p}}\cdot n^{1-2/p}\log\log n}$ indices of possible heavy-hitters in turnstile streams, $\countsketch$ uses at most $\O{\frac{1}{\eps^{4/p}}\cdot n^{1-2/p}\log^2 n\log\log n}$ space. 
By the same reasoning, we use space $\O{\frac{1}{\eps^{4/p}}\cdot n^{1-2/p}\log n\log\log n}$ in insertion-only streams by using the more space efficient $\bptree$. 

\begin{restatable}{theorem}{thmtwopassinsertion}
\thmlab{thm:twopass:insertion}
For $p>2$, there exists a two-pass insertion-only streaming algorithm that outputs a $(1+\eps)$-approximation to the $F_p$ moment with probability at least $\frac{2}{3}$, while using $\O{\frac{1}{\eps^{4/p}}\cdot n^{1-2/p}\log n\log\log n}$ bits of space.
\end{restatable}
\begin{proof}
The proof of correctness is exactly the same as that of \thmref{thm:twopass:turnstile} since using $\bptree$ as the subroutine for $\heavyhitters$ rather than $\countsketch$ offers the same guarantee for insertion-only streams. 
By \thmref{thm:bptree}, $\bptree$ with threshold $\eps$ requires $\O{\frac{1}{\eps^2}\log n}$ bits of space to output the indices of the heavy-hitters. 
To analyze the space complexity, note that \algref{alg:multipass:estimator} runs a separate instance of $\bptree$ with threshold $\frac{(\eps_i)^{2/p}}{80\gamma}\cdot\left(\frac{1}{n_i}\right)^{1/2-1/p}$ to output a set $H^r_i$ for $r\le\O{\log\log n}$, where $\gamma$ is a sufficiently large constant. 
Hence, the first pass across all indices $i$ for a fixed $r$ uses space at most
\[\left(C_1\log n\right)\cdot\sum_{i\in[\alpha\log n]}\frac{(n_i)^{1-2/p}}{(\eps_i)^{4/p}},\]
for some absolute constants $C_1,\alpha>0$. 
Since $n_i=\frac{10\gamma n}{2^i}$ and $\eps_i=\frac{\eps}{16\eta\cdot 2^{i(p/16-1/8)}\log\frac{1}{\eps^2}}$ for a sufficiently large constant $\eta$, the first pass uses space at most
\[\left(C_1\log n\right)\cdot\sum_{i=1}^\infty\frac{C_2n^{1-2/p}}{\eps^{4/p}}\cdot\frac{2^{i(1/4-1/(2p))}}{2^{i(1-2/p)}},\]
for some absolute constants $C_1,C_2>0$, across all indices $i$ for a fixed $r$. 
As $\sum_{i=1}^\infty\frac{2^{i(1/4-1/(2p))}}{2^{i(1-2/p)}}=\sum_{i=1}^\infty\frac{1}{2^{3i(1-2/p)/4}}$ is a geometric series that is upper bounded by some constant, then the total space in the first pass is $\O{\frac{1}{\eps^{4/p}}\cdot n^{1-2/p}\log n}$ across all indices $i$ for a fixed $r$. 
Since $r\le\O{\log\log n}$, then the total space in the first pass is $\O{\frac{1}{\eps^{4/p}}\cdot n^{1-2/p}\log n\log\log n}$. 

The second pass tracks the frequencies of each item reported by some instance of $\bptree$. 
Since at most $\O{\frac{1}{\eps^{4/p}}\cdot n^{1-2/p}\log\log n}$ indices can be reported across all instances of $\bptree$ and $\O{\log n}$ bits of space can be used to track the frequency of each reported index, then the total space for the second pass is at most $\O{\frac{1}{\eps^{4/p}}\cdot n^{1-2/p}\log n\log\log n}$.  
Thus, the total space is $\O{\frac{1}{\eps^{4/p}}\cdot n^{1-2/p}\log n\log\log n}$.
\end{proof}

\section{Lower Bounds}
\seclab{sec:lb}
In this section, we first consider the standard \emph{blackboard} communication model, where a number of players each have a local input and the goal is to solve some predetermined communication problem by sending messages to a shared medium. 
Each player is assumed to have access to an unlimited amount of private randomness. 
The sequence of messages on the shared blackboard is called the \emph{transcript} and the maximum length of the transcript over all inputs is the \emph{communication cost} of a given protocol. 
The communication complexity of $f$, denoted by $R_{\delta}(f)$, is the minimal communication cost of all protocols that succeed with probability at least $1-\delta$ for all legal inputs to $f$. 
We now require a number of basic concepts and results from information theory. 
\begin{definition}[Mutual information]
Given a pair of random variables $X$ and $Y$ with joint distribution $p(x,y)$, the \emph{mutual information} is defined as $I(X;Y):=\sum_{x,y}p(x,y)\log\frac{p(x,y)}{p(x)p(y)}$, for marginal distributions $p(x)$ and $p(y)$.  
\end{definition}

\begin{definition}[Information cost]
Let $\Pi$ be a randomized protocol that produces a random variable $\Pi(X_1,\ldots,X_t)$ as a transcript on inputs $X_1,\ldots,X_t$ drawn from a distribution $\mu$. 
Then the \emph{information cost} of $\Pi$ with respect to $\mu$ is defined as $I(P_1,\ldots,P_t;\Pi(P_1,\ldots,P_t))$.  
\end{definition}

\begin{definition}[Information complexity]
The information complexity of $f$ with respect to a distribution $\mu$ and failure probability $\delta$ is the minimum information cost of a protocol for $f$ with respect to $\mu$ that fails with probability at most $\delta$ on every input and denoted by $\IC_{\mu,\delta}(f)$. 
\end{definition}

\begin{definition}[Conditional information cost]
Suppose $((X_1,\ldots,X_t),W)\sim\zeta$ for a mixture of product distributions. 
The \emph{conditional information cost} of a randomized protocol $\Pi$ with respect to $\zeta$ is $I(X_1,\ldots,X_t;\Pi(X_1,\ldots,X_t)|W)$. 
\end{definition}

\begin{definition}[Conditional information complexity]
The conditional information complexity of $f$ with respect to $\zeta$ and failure probability $\delta$ is defined to be the minimum conditional information cost of a protocol for $f$ with respect to $\zeta$ that fails with probability at most $\delta$ and denoted by $\CIC_{\zeta,\delta}(f)$. 
\end{definition}

\begin{fact}
\factlab{fact:cc:info}
For any distribution $\mu$ and failure probability $\delta\in(0,1)$, we have $R_{\delta}(f)\ge\IC_{\mu,\delta}(f|W)$. 
\end{fact}

\begin{definition}[Decomposable functions]
A function $f$ is $g$-decomposable with primitive $h$ if $f(\x_1,\ldots,\x_t)=g(h(\x_{1,1},\ldots,\x_{t,1}),\ldots,h(\x_{1,n},\ldots,\x_{t,n}))$ for some function $h$, where $\x_{i,j}$ denotes the $j$-th coordinate of $\x_i$. 
\end{definition}

\begin{definition}[Embedding of a coordinate into a vector]
For a vector $\x\in\mathbb{R}^n$, $i\in[n]$ and $y\in\mathbb{R}$, we use $\embed(\x,i,y)$ to denote the vector $\v$ such that $\v_i=y$ and $\v_j=\x_j$ for $j\neq i$. 
\end{definition}

\begin{definition}[Collapsing distribution]
For a function $f$ that is $g$-decomposable with primitive $h$, we say $(\x_1,\ldots,\x_t)$ is a collapsing input for $f$ if for every $j$ and set $(y_1,\ldots,y_t)$ of legal inputs, we have
\[f(\embed(\x_1,j,y_1),\ldots,\embed(\x_t,j,y_t))=h(y_1,\ldots,y_t).\]
A distribution $\mu$ on the input space is a \emph{collapsing distribution} for $f$ if every $(\x_1,\ldots,\x_t)$ in the support of $\mu$ is a collapsing input. 
\end{definition}

\begin{lemma}[Information cost decomposition lemma, Lemma 5.1 in~\cite{Bar-YossefJKS04}]
\lemlab{lem:decompose}
Let $\mu$ be a mixture of product distributions and suppose that $((X_1,\ldots,X_t),W)\sim\mu^n$. 
Then for a protocol $\Pi$ whose inputs $(X_1,\ldots,X_t)\sim\mathcal{L}^n$ for some legal input set $\mathcal{L}$, we have $I(X_1,\ldots,X_t;\Pi(X_1,\ldots,X_t|W))\ge\sum_i I(X_{1,i},\ldots,X_{t,i};\Pi(X_1,\ldots,X_t)|W)$, where $X_{i,j}$ denotes the $j$-th component of $X_i$. 
\end{lemma}

\begin{lemma}[Reduction lemma, Lemma 5.5 in~\cite{Bar-YossefJKS04}]
\lemlab{lem:reduction}
Let $\Pi$ be a protocol for a decomposable function $f$ with input $\mathcal{L}^n$ and primitive $h$, that fails with probability at most $\delta$. 
Let $\mu$ be a mixture of product distributions and suppose $((X_1,\ldots,X_t),D)\sim\mu^n$. 
If $(X_1,\ldots,X_t)$ is a collapsing distribution for $f$, then for all $j\in[n]$,
\[I((X_{1,j},\ldots,X_{t,j});\Pi(X_1,\ldots,X_t)|D)\ge\CIC_{\mu,\delta}(h).\]
\end{lemma}


For an input $X$ to a protocol $\Pi$, let $\pi(X)$ denote the probability distribution over the transcripts produced by $\Pi$ across the private randomness of the protocol. 
Let $\psi(X):=\sqrt{\pi(X)}$ denote the transcript wave function of $X$ so that $\|\psi(X)\|_2=\|\pi(X)\|_1=1$. 

\begin{definition}
The Hellinger distance between $\psi_1$ and $\psi_2$ is the scaled Euclidean distance defined by $h(\psi_1,\psi_2):=\frac{1}{\sqrt{2}}\|\psi_1-\psi_2\|_2$. 
\end{definition}

\begin{fact}[Mutual information to Hellinger distance]
\factlab{fact:info:hellinger}
For $X\in\{u,v\}$, we have $I(X;\Pi)\ge\frac{1}{2}\|\psi(u)-\psi(v)\|_2^2$. 
\end{fact}

\begin{fact}[Soundness for Hellinger distance]
\factlab{fact:soundness:hellinger}
If $\Pi$ is a protocol for $g$ that fails with probability at most $\delta\in(0,1)$ and $g(X_1)\neq g(X_2)$, then
\[\frac{1}{2}\|\psi(X_1)-\psi(X_2)\|_2^2\ge1-2\sqrt{\delta}.\]
\end{fact}

\begin{lemma}[Theorem 7 in~\cite{Jayram09}]
\lemlab{lem:disjoint:char}
For a $t$-party protocol $\Pi$ on the input domain $\{0,1\}^t$, let $A_1,\ldots,A_s$ be a pairwise disjoint collection of $s=2^k$ subsets of $[t]$ and $A=\cup_i A_i$. 
Then
\[\sum_{i=1}^s\|\psi(\emptyset)-\psi(A_i)\|_2^2\ge\|\psi(\emptyset)-\psi(A)\|_2^2\cdot\prod_{\ell=1}^k\left(1-\frac{1}{2^\ell}\right).\]
\end{lemma}

\begin{definition}
In the $(t,\eps,n)$-player set disjointness estimation problem $(t,\eps,n)-\disjinfty$, there are $t+1$ players $P_1,\ldots,P_{t+1}$ with private coins in the standard blackboard model. 
For $s\in[t]$, each player $P_s$ receives a vector $\v_s\in\{0,1\}^n$ and player $P_{t+1}$ receives both an index $j\in[n]$ and a bit $c\in\{0,1\}$. 
For $\u=\sum_{s\in[t]}\v_s$, the inputs are promised to satisfy $u_i\le 1$ for each $i\neq j$ and either $u_j=1$ or $u_j=t$. 
With probability at least $\frac{9}{10}$, $P_{t+1}$ must differentiate between the three possible input cases:
\begin{enumerate}
\item $u_j+\frac{ct}{\eps}\le t$
\item $u_j+\frac{ct}{\eps}\in\left\{\frac{t}{\eps},\frac{t}{\eps}+1\right\}$
\item $u_j+\frac{ct}{\eps}=(1+\eps)\frac{t}{\eps}$,
\end{enumerate}
where $\eps\in(0,1)$. 
We call coordinate $j\in[n]$ the \emph{spike} location. 
\end{definition}

\textbf{Direct sum embedding.} 
We use the direct sum technique by showing that even in the case where $u_i\le 1$ for all $i\in[n]$, the information cost of the players is sufficiently high. 
We describe the embedding performed by each player to sample coordinates independently conditioned on the auxiliary variable $D$. 
We define the auxiliary variable $D=(D_1,\ldots,D_n)$ by first defining a distribution $\mu$ for $D_i$ and $v_{s,i}$, for each fixed coordinate $i\in[n]$ and across all players $s\in[t]$:
\begin{enumerate}
\item
$D_i\sim[t]$ uniformly at random, so that $D_i$ determines a player whose input bit will be randomized while the remaining players have input bit zero.
\item
Conditioned on $D_i=s$, then each player $P_r$ with $r\neq s$ sets $v_{r,i}=0$ while player $P_s$ independently and uniformly sets $v_{s,i}\in\{0,1\}$. 
\end{enumerate} 
We define the distribution $\zeta:=\mu^n$ so that the auxiliary random variable $D=(D_1,\ldots,D_n)$ is the vector whose $i$-th coordinate determines the player $P_{D_i}$ whose $i$-th bit $v_{D_i,i}$ in their input vector will vary, while the other players $P_s$ have $i$-th bit $v_{s,i}$ zero in their vectors, for $s\neq D_i$. 
Observe that under the distribution $\zeta$, we indeed have $u_i\le 1$ for each coordinate $i\in[n]$ of $\u=\sum_{s\in[t]}\v_s$. 
Finally, we fix the input $c=0$ to player $P_{t+1}$ and pick $j\in[n]$ according to any arbitrary distribution.  

We now show that $\zeta$ is a collapsing distribution for $(t,\eps,n)-\disjinfty$.  
\begin{restatable}{lemma}{lemzetacollapsing}
\lemlab{lem:zeta:collapsing}
For $c=0$ and arbitrary $j\in[n]$, $\zeta$ is a collapsing distribution for $(t,\eps,n)-\disjinfty$.
\end{restatable}
\begin{proof}
The distribution $\zeta$ only places mass at $0^t$ or the elementary vector $\e_{D_i}\in\{0,1\}^t$ for each $i\in[n]$. 
Since $c=0$, then for every input in the support of $\zeta$, we have
\[(t,\eps,n)-\disjinfty(\embed(\v_1,i,y_1),\ldots,\embed(\v_t,i,y_t))=\disjinfty(y_1,\ldots,y_t),\]
regardless of the value of $j$. 
Hence, $\zeta$ is a collapsing distribution for $(t,\eps,n)-\disjinfty$.
\end{proof}

We now prove the communication complexity of the $(t,\eps,n)$-player set disjointness estimation problem. 
\begin{restatable}{theorem}{thmdisjinfty}
\thmlab{thm:disjinfty}
The total communication complexity for the $(t,\eps,n)$-player set disjointness estimation problem is $\Omega\left(\frac{n}{t}\right)$. 
\end{restatable}
\begin{proof}
Let $\Pi$ be a protocol with minimum information cost that solves any input to $(t,\eps,1)-\disjinfty$ with probability at least $1-\delta$, where $\delta=\frac{1}{10}$. 
By setting $c=0$ and choosing $j\in[n]$ from any arbitrary distribution for the input to $P_{t+1}$, we have that $\zeta$ is a collapsing distribution for $(t,\eps,n)-\disjinfty$, by \lemref{lem:zeta:collapsing}. 
$\zeta$ also induces the auxiliary random variable $D=(D_1,\ldots,D_n)$ that determines the player $P_{D_i}$ whose $i$-th bit $v_{D_i,i}$ in their input vector will vary. 
 
We note that the messages from each player $P_s$ with $s\in[t]$ do not depend on the input of $P_{t+1}$. 
Moreover, the input of player $P_{t+1}$ is drawn from a distribution that is independent of the incoming message. 
Since our protocol is required to be correct on all inputs as opposed to some particular distribution, then it suffices to consider the protocol $\Pi$ only on the input of the first $t$ players, which we denote by $\Pi(\v_1,\ldots,\v_t)$. 
Thus by applying the direct sum technique on the embedding $\zeta=\mu^n$, we have from \factref{fact:cc:info} that 
\[R((t,\eps,n)-\disjinfty)\ge\IC_{\zeta,\delta}(\v_1,\ldots,\v_t;\Pi(\v_1,\ldots,\v_t|D)).\]
Since $\IC_{\zeta,\delta}(\v_1,\ldots,\v_t;\Pi(\v_1,\ldots,\v_t|D))\ge\sum_{i=1}^nI(v_{1,i},\ldots,v_{t,i};\Pi(v_{1,i},\ldots,v_{t,i})|D)$ by the decomposition of information cost in \lemref{lem:decompose}, we have that
\[R((t,\eps,n)-\disjinfty)\ge\sum_{i=1}^nI(v_{1,i},\ldots,v_{t,i};\Pi(v_{1,i},\ldots,v_{t,i})|D).\]
Notably, \lemref{lem:decompose} implicitly uses the chain rule for conditional entropy and the subadditivity of conditional entropy. 

\textbf{Decomposition into single coordinate problem.} 
Since $\zeta$ is a collapsing distribution for $c=0$ by \lemref{lem:zeta:collapsing}, then for the single coordinate problem $(t,\eps,1)-\disjinfty$ of $(t,\eps,n)-\disjinfty$,so that all vector inputs to each player have dimension one, we have 
\[\sum_{i=1}^nI(v_{1,i},\ldots,v_{t,i};\Pi(v_{1,i},\ldots,v_{t,i})|D)\ge n\cdot\CIC_\mu((t,\eps,1)-\disjinfty),\] 
by \lemref{lem:reduction}. 
Thus we can lower bound the communication complexity of $(t,\eps,n)-\disjinfty)$ by the single coordinate problem through
\[R((t,\eps,n)-\disjinfty)\ge n\cdot\CIC_\mu((t,\eps,1)-\disjinfty).\]
In particular, we have that the input distribution $\zeta$ for $(t,\eps,n)-\disjinfty$ has single coordinate distribution $\mu$ for $(t,\eps,1)-\disjinfty$. 
Thus it remains to show that $\CIC_\mu((t,\eps,1)-\disjinfty)=\Omega\left(\frac{1}{t}\right)$. 

\textbf{Complexity of single coordinate problem.} 
Without loss of generality, we consider the single coordinate problem by considering the first coordinates $v_{1,1},\ldots,v_{t,1}$ of the input vectors $\v_1,\ldots,\v_t$. 
By the definition of conditional information complexity, we have
\[\CIC_\mu((t,\eps,1)-\disjinfty)=I(v_{1,1},\ldots,v_{t,1};\Pi|D_1)\ge\frac{1}{t}\sum_{s\in[t]}I(v_{1,1},\ldots,v_{t,1};\Pi|D_1=s).\]
Conditioned on $D_1=s$, we have $v_{s,1}\in\{0,1\}$ uniformly at random and $v_{r,1}=0$ for $r\neq s$. 
Thus by \factref{fact:info:hellinger} relating mutual information to Hellinger distance, we have
\[\CIC_\mu((t,\eps,1)-\disjinfty)\ge\frac{1}{t}\sum_{s=1}^t\frac{1}{2}\|\psi(\emptyset)-\psi(\{s\})\|_2^2,\]
where we use $\emptyset$ to denote that $v_{s,1}=0$ and $\{s\}$ to denote that $v_{s,1}=1$. 
Instead of analyzing the sum of each of the squared Hellinger distances, we instead use \lemref{lem:disjoint:char} to analyze the squared Hellinger distance from $\emptyset$ to $[t]$, which represents that $v_{1,1}=\ldots=v_{t,1}=1$. 
For $t=2^k$, we set the parameter $s$ in \lemref{lem:disjoint:char} to $t$ and $A_i=\{i\}$, so that
\[\CIC_\mu((t,\eps,1)-\disjinfty)\ge\frac{1}{t}\left(\frac{1}{2}\|\psi(\emptyset)-\psi([t])\|_2^2\right)\cdot\prod_{\ell=1}^k\left(1-\frac{1}{2^\ell}\right).\]
Since $(t,\eps,1)-\disjinfty$ on inputs $\emptyset$ and $[t]$ achieve different outputs, then by \factref{fact:soundness:hellinger}, 
\[\frac{1}{2}\|\psi(\emptyset)-\psi([t])\|_2^2\ge1-2\sqrt{\delta},\]
where $\delta$ is the probability of failure. 
Similarly, we have $\prod_{\ell=1}^k\left(1-\frac{1}{2^\ell}\right)\ge\prod_{\ell=1}^\infty\left(1-\frac{1}{2^\ell}\right)\approx 0.28878$, by the definition of the digital search tree constant. 
Thus, we have that $\CIC_\mu((t,\eps,1)-\disjinfty)=\Omega\left(\frac{1}{t}\right)$, which implies $R((t,\eps,n)-\disjinfty)=\Omega\left(\frac{n}{t}\right)$.
\end{proof}

We remark that the communication complexity of \thmref{thm:disjinfty} matches the communication complexity of $t$-player set disjointness. 
However, since the problem requires correctness on all inputs, then we can distinguish between the possible input cases by focusing on the specific coordinate $j$ given to player $P_{t+1}$. 
By contrast, the reduction of \cite{Ganguly12} from $t$-player set disjointness requires an algorithm to ``test'' all coordinates $i\in[n]$ for the spike location. 
Thus the reduction requires that an $F_p$ moment estimation algorithm succeeds with probability $1-\frac{1}{\poly(n)}$, thereby losing a multiplicative $\O{\log n}$ factor and achieving $\Omega\left(\frac{n^{1-2/p}}{\eps^2 \log n} \right)$ in the space lower bound for $F_p$ moment estimation.  
Since our communication problem gives the specific spike location as input, our reduction only requires an $F_p$ moment estimation algorithm that succeeds with constant probability.
We remark that a similar technique was used in \cite{LiW13} for the $L_\infty$ promise problem. 

\textbf{Reduction.} 
Let $t=\Theta\left(\frac{1}{\eps}\cdot n^{1/p}\right)$. 
We reduce $(1+\eps)$-approximate $F_p$ moment estimation to an instance of $(t,\eps)-\disjinfty$ as follows. 
Let $\calA$ be any fixed randomized one-pass insertion-only streaming algorithm that outputs a $\left(1+\frac{\eps}{3}\right)$-approximation to the $F_p$ moment of the underlying frequency vector with probability at least $\frac{2}{3}$. 
Recall that the first $t$ players each receive a vector $\v_s$ with $s\in[t]$ and player $P_{t+1}$ receives both a special index $j\in[n]$ and a bit $c\in\{0,1\}$. 
\begin{itemize}
\item
For each $s\in[t]$, player $P_s$ takes the state of the algorithm $\calA$, inserts the coordinates of vector $\v_s$, and passes the state of the algorithm to player $P_{s+1}$. 
\item
Player $P_{t+1}$ takes the state of $\calA$, adds the vector $\frac{ct}{\eps}\cdot\e_j$, where $\e_j$ is the elementary vector with a one in position $j$ and zeros elsewhere, and then queries the state of $\calA$ to obtain a $\left(1+\frac{\eps}{3}\right)$-approximation to the $F_p$ moment of underlying frequency vector. 
\end{itemize}
\begin{restatable}{theorem}{thmlb}
\thmlab{thm:lb}
For any constant $p>2$ and parameter $\eps=\Omega\left(\frac{1}{n^{1/p}}\right)$, any one-pass insertion-only streaming algorithm that outputs a $(1+\eps)$-approximation to the $F_p$ moment of an underlying frequency vector with probability at least $\frac{9}{10}$ requires $\Omega\left(\frac{1}{\eps^2}\cdot n^{1-2/p}\right)$ bits of space. 
\end{restatable}
\begin{proof}
Player $P_{t+1}$ receives the state of $\calA$ on the input $\u=\sum_{s\in[t]}\v_s$ and induces a frequency vector $\x:=\u+\frac{ct}{\eps}\cdot\e_j$. 
Recall that the inputs are promised to satisfy $u_i\le 1$ for each $i\neq j$ and either $u_j=1$ or $u_j=t$. 
If $c=0$, then $\x=\u$ so that for a constant $C>0$ and $t=\frac{C}{\eps}\cdot n^{1/p}$,  
\[\|\x\|^p_p\le F_0+t^p\le n+\frac{C^p}{\eps^p}\cdot n.\]
If $c=1$ and $u_j\le 1$, then for $t=\frac{C}{\eps}\cdot n^{1/p}$, we have $\|\x\|^p_p\ge\left(\frac{t}{\eps}\right)^p=\frac{C^p}{\eps^{2p}}\cdot n$ and thus
\[\|\x\|^p_p\le F_0+\left(1+\frac{t}{\eps}\right)^p\le n+p+\frac{p\cdot C^p}{\eps^{2p}}\cdot n.\]
Finally, if $c=1$ and $u_j=t$, then 
\[\|\x\|^p_p\ge\left(t+\frac{t}{\eps}\right)^p=\left(1+\frac{1}{\eps}\right)^p\cdot\frac{C^p}{\eps^{2p}}\cdot n.\]
For sufficiently small $\eps\in(0,1)$ with $\eps=\Omega\left(\frac{1}{n^{1/p}}\right)$ and constant $p>2$, there exists a constant $C>0$ such that these three cases are separated by a multiplicative $(1+\eps)$. 
Since $\calA$ outputs a $\left(1+\frac{\eps}{3}\right)$-approximation to the $F_p$ moment of the underlying frequency vector, then player $P_{t+1}$ obtains a $\left(1+\frac{\eps}{3}\right)$-approximation to $\|\x\|_p^p$ and can differentiate between the three cases, thus solving $(t,\eps)-\disjinfty$ with probability at least $\frac{9}{10}$. 
By \thmref{thm:disjinfty}, $\calA$ uses $\Omega\left(\frac{n}{t}\right)$ total communication across the $t+1$ players. 
Therefore, the space complexity of $\calA$ is $\Omega\left(\frac{n}{t^2}\right)=\Omega\left(\frac{1}{\eps^2}\cdot n^{1-2/p}\right)$ for $t=\Theta\left(\frac{1}{\eps}\cdot n^{1/p}\right)$. 
The result then follows from rescaling $\eps$. 
\end{proof}

\section*{Acknowledgements}
The authors would like to thank support from NSF grant No. CCF-181584 and a Simons Investigator Award and thank Krzysztof Onak for pointing toward the reference \cite{AndoniMOP08}.

\def\shortbib{0}
\bibliographystyle{alpha}
\bibliography{references}

\appendix
\section{Approximately Counting Heavy-Hitters}
\applab{sec:counthh}
In this section, we describe the procedure $\counthh$ introduced by \cite{BravermanGW20} for approximating the frequencies of the $L_2$-heavy hitters in a random-order insertion-only stream. 
Their algorithm partitions the updates of a random-order stream into blocks of updates. 
They then randomly choose a number of different coordinates from the universe $[n]$ to test in each block. 
If there is an update to a coordinate that is tested within a block, then the coordinate is said to be \emph{excited}. 
Once an item is excited, the algorithm then explicitly tracks updates to the item across a number of subsequent blocks. 
If the algorithm sees that the stream consistently updates a tracked item, then the item is promoted to a heavy-hitter, and its frequency remains tracked over a longer number of blocks. 
The frequency over the longer number of blocks is then scaled up to provide an accurate estimation of the frequency of the heavy-hitter in the entire stream. 
Otherwise if a tracked item is not consistently updated by the stream, then it was likely a false positive due to the randomness of the stream and the hash functions, and it is no longer tracked by the algorithm. 
Any heavy-hitter on the stream has sufficiently high probability of being updated in a block in which it is tested. 
Thus all heavy-hitters will become excited with ``good'' probability and become tracked by the algorithm. 

The algorithm heavily relies on the fact that the stream is random-order and insertion-only. 
In streams that are not insertion-only, any initially reported heavy-hitter can simply be erased later in the stream. 
In streams that are not random-order, the distribution of the heavy-hitters may not be uniform, so they may not become excited in a block in which they are tested. 
Similarly, the distribution of the heavy-hitters in the blocks for which they are tracked must be representative of the overall frequency of the heavy-hitters due to the random-order model. 

\cite{BravermanGW20} showed that $\O{\frac{1}{\eps^2}}$ items are tracked at a time. 
We wish to report an approximate frequency to a heavy-hitter with accuracy $(1+\eps)$.
Thus, each heavy-hitter is tracked for a small number of blocks, so that the counter for the heavy-hitter across these blocks can be represented using $\O{\log\frac{1}{\eps}+\log\log n}$ bits. 
However, the number of blocks is sufficiently large so that the tracking provides a $(1+\eps)$-approximation to the frequency of the item across the entire stream, due to the uniformity properites of the random-order stream. 

We first give the subroutine $\counthhsub$ from \cite{BravermanGW20} for approximating the number of heavy-hitters in \algref{alg:counthh:sub}. 

\begin{algorithm}[!htb]
\caption{$\counthhsub$: Subroutine for counting the number of heavy-hitters \cite{BravermanGW20}}
\alglab{alg:counthh:sub}
\begin{algorithmic}[1]
\Require{Random order $a_1,\ldots,a_m\in[n]$ of stream of updates to coordinates of an underlying frequency vector, threshold parameter $\eps\in(0,1)$, stream length $m$ and $2$-approximation $\tilde{F_2}$ to $F_2$.}
\Ensure{Estimated frequencies to $L_2$-heavy hitters.}
\State{$H\gets\emptyset$}
\State{Partition the stream into $t=\eps\sqrt{\tilde{F_2}}$ contiguous blocks, $B_1,\ldots,B_t$, each of length $\frac{F_1}{\eps\sqrt{\tilde{F_2}}}$.}
\State{Let $h^1,\ldots,h^{10\ln(1/\eps)}:[n]\to[t]$ and $g:[n]\to[200\eps^{-26}\ln^4(dm)]$ be independent pairwise hash functions.}
\While{processing block $B_i$}
\State{$S^i\gets\emptyset$, $e^i=0$, initialize $c^i$ to the all $0$'s vector.}
\State{Process all items $a_j\in B_i$.}
\While{processing item $a_j$}
\If{$e^i=0$ and $h^q(a_j)=i$ for some $q\in[10\ln(1/\eps)]$}
\State{$S^i\gets S^i\cup\{g(a_j)\}$.}
\EndIf
\If{$|S^i|>100\eps^{-2}\ln(1/\eps)$}
\State{$e^i\gets 1$}
\EndIf
\If{$g(a_j)\in S^{i-1}$, $h^q(a_j)=i-1$ for some $q\in[10\ln(1/\eps)]$ and $e^{i-1}=0$}
\State{$c^{i-1}(g(a_j))=c^{i-1}(g(a_j))+1$}
\EndIf
\EndWhile
\If{$e^{i-1}=0$ and there is a unique $k\in S^{i-1}$ with $c^{i-1}(k)=1$}
\If{there is no uncompleted $\identify$ instance already running and $|H|\le2\eps^{-2}$}
\State{$H\gets H\cup\identify(\{h^q\}_q,g,i,k,m,a_{iF_1/(\eps\sqrt{\tilde{F_p}})+1},\ldots,a_m)$}
\EndIf
\EndIf
\EndWhile
\State{\Return $|H|$}
\end{algorithmic}
\end{algorithm}

\begin{algorithm}[!htb]
\caption{Algorithm for identifying heavy-hitters \cite{BravermanGW20}}
\alglab{alg:identify}
\begin{algorithmic}[1]
\Require{Hash functions $\{h^q\}_q$, hash function $g$, block number $i$, tracking hash value $k$, stream length $m$, stream updates $a_{iF_1/(\eps\sqrt{\tilde{F_p}})+1},\ldots,a_m$.}
\Ensure{Set of heavy-hitters}
\State{$\id\gets\emptyset$, $\ell\gets 1$}
\While{$\id\neq\emptyset$ or $\ell>100\ln(1/\eps)$}
\If{$|\{j\in B_{i+\ell}|g(j)=k\text{ and }h^q(j)=i-1\text{ for some }q\in[10\ln(1/\eps)]\}|=1$}
\State{$\id\gets j$}
\Else
\State{$\ell\gets\ell+1$}
\EndIf
\EndWhile
\If{$\id=\emptyset$}
\State{\Return $\emptyset$}
\EndIf
\State{$p=4000\eps^{-2}\ln m$}
\State{Let $\xi$ be the number of occurrences of $\id$ in the multiset $\uplus_{v=1}^p B_{i+\ell+v}$.}
\If{$\xi\ge(1.01/\sqrt{2})p$}
\State{$\widehat{f_{\id}}=\frac{\eps\sqrt{\tilde{F_p}}}{p}\cdot\xi$}
\State{\Return $(\id, \widehat{f_{\id}})$}
\Else
\State{\Return $\emptyset$}
\EndIf
\end{algorithmic}
\end{algorithm}

\begin{algorithm}[!htb]
\caption{$\counthh$: Algorithm for counting the number of heavy-hitters \cite{BravermanGW20}}
\alglab{alg:counthh}
\begin{algorithmic}[1]
\Require{Random order $a_1,\ldots,a_m\in[n]$ of stream of updates to coordinates of an underlying frequency vector, threshold parameter $\eps\in(0,1)$, stream length $m$ and $2$-approximation $\tilde{F_p}$ to $F_p$.}
\Ensure{Estimated frequencies to $L_2$-heavy hitters.}
\State{Let $\widehat{F_2^{(t)}}$ be the output of an $F_2$-sketch at time $t$.}
\For{each update $a_t$ with $t\in[m]$}
\If{$t=2^j-1$ for some integer $j\ge 0$}
\State{Let $H_j$ be the output of $\counthhsub_j$ with parameter $\frac{\eps}{10}$.}
\State{Initialize instance $\counthhsub_{j+1}(1.01\cdot\widehat{F_2^{(t)}},2^{j+1},a_{t+1},\ldots,a_{t+2^{j+1}})$ with parameter $\frac{\eps}{10}$.}
\State{Discard $H_{j-1}$ and $\counthhsub_{j-1}$.}
\EndIf
\EndFor
\State{Let $j$ be the largest index with an instance $\counthhsub_j$.}
\State{\Return $\left(i,\frac{m}{2^j},\widehat{x_i}\right)$ in $H_{j-1}$ for which $\frac{m}{2^j}\cdot\widehat{x_i}\ge\left(1-\frac{\eps}{5}\right)\eps\cdot\frac{\sqrt{\widehat{F_2^{(m)}}}}{1.01}$.}
\end{algorithmic}
\end{algorithm}

We now recall a number of key properties shown in \cite{BravermanGW20}. 


\begin{lemma}[No false positives by $\identify$, Lemma 4 in~\cite{BravermanGW20}]
\lemlab{lem:identify:no:small}
With probability at least $1-1/m$, no instance of $\identify$ reports an element $j\neq\emptyset$ with $f_j^2\le\frac{\eps^2}{2}\cdot F_2$.
\end{lemma}

\begin{corollary}[$|H|$ is small, Corollary 29 in~\cite{BravermanGW20}]
\corlab{cor:small:H}
With probability at least $1-1/m$, we have $|H|\le 2\eps^{-2}$ at all times during the stream.
\end{corollary}

\begin{lemma}[Heavy-hitters in consecutive blocks, Lemma 5 in~\cite{BravermanGW20}]
Let $HH$ be the set of items $j$ with $f_j^2\ge\eps^2\cdot F_p$.
With probability at least $1-1/\sqrt{m}$, we simultaneously have for all $j\in HH$ and every sequence of $4000\eps^{-2}\ln m$ consecutive blocks, $B_{i+1},\ldots,B_{i+4000\eps^{-2}\ln m}$, that coordinate $j$ appears in at least $4000(1.01/\sqrt{2})\eps^{-2}\ln m$ of these blocks.	
\end{lemma}

\begin{lemma}[Accuracy of heavy-hitter estimations, Lemma 6 in~\cite{BravermanGW20}]
\lemlab{lem:identify:accuracy}
With probability at least $1-1/m^{98}$, we simultaneously have for all $j$ with $f_j^2\ge\frac{\eps^2}{2}\cdot F_2$ and every sequence of $p:=4000\eps^{-2}\ln m$ consecutive blocks, the number of occurrences $\xi$ of $j$ in $\uplus_{v=1}^p B_{i+\ell+v}$ satisfies $\xi\in(1\pm\eps)\cdot\frac{p\cdot f_j}{\eps\sqrt{\tilde{F_2}}}$. 
\end{lemma}

\begin{lemma}[Identification of heavy-hitters, Lemma 7 in~\cite{BravermanGW20}]
\lemlab{lem:identify:all:hh}
With probability at least $39/40$, simultaneously for all $j\in HH$, there exists some $q\in[10\ln(1/\eps)]$ such that:
\begin{enumerate}
\item
$h^q(j)=i$ with $e^i=0$ and $g(j)\in S^i$ after $B_i$ is processed
\item
$g(j)$ is the unique value $k\in S^i$ with $c^i(k)=1$
\item
No instance of $\identify$ is already running and $|H|\le2\eps^{-2}$.
\end{enumerate}
\end{lemma}

Together, these statements give the following guarantees:
\thmcounthh*
\end{document}